%% file: main_ms.tex
\documentclass[12pt]{article}

\usepackage[T1]{fontenc}
\usepackage[english]{babel}
\usepackage[letterpaper,margin=1.25in]{geometry}
\usepackage{setspace}

\usepackage{amsmath,amssymb,amsfonts,amsthm,mathtools}

\DeclareMathOperator{\MRCI}{MRCI}

\DeclareMathOperator{\E}{\mathbb{E}}

\usepackage{graphicx}
\usepackage{booktabs}
\usepackage[flushleft]{threeparttable}
\usepackage{subcaption}
\usepackage{tabularx}
\graphicspath{{./}}

\usepackage{natbib}
\usepackage[colorlinks,citecolor=black,urlcolor=black,linkcolor=black]{hyperref}
\usepackage{hypernat}

\usepackage[dvipsnames]{xcolor}
\usepackage{enumitem}

\newtheorem{proposition}{Proposition}

\newtheorem{definition}{Definition}
\newtheorem{assumption}{Assumption}

\begin{document}

\title{How Many Mechanisms? Measuring Parsimony in Risky Choice}
\author{Avner Seror\footnote{Aix-Marseille University, CNRS, AMSE, Marseille, France. Email: avner.seror@univ-amu.fr. I am grateful to Marina Agranov, Benjamin Enke, En Hua Hu, Dotan Persitz, John Quah, and Thierry Verdier for helpful discussions. I acknowledge funding from the French government under the ANR
JCJC “BIAS" project (reference: ANR-25-CE26-7164-01) and the “France 2030” investment plan managed by the French National Research Agency (reference: ANR-17-EURE-0020) and from the Excellence Initiative of Aix-Marseille University – A*MIDEX. Mistakes are my own. Refine.ink was used to check the paper for consistency and clarity. All scientific content, modeling choices, and results were produced and verified by the author. Pre-registration exemption: this paper does not generate new data. The empirical section applies the framework to three publicly available datasets. }} 
\date{May, 2026}
\maketitle
\onehalfspacing
\thispagestyle{empty}

\begin{abstract}
Behavioral theories rest on parsimony: a small number of mechanisms organizing many decisions. We define a Maximum Rule Concentration Index that measures how parsimoniously a dataset of risky choices can be organized through a library of simple, parameter-free decision rules drawn from canonical behavioral theories: salience, regret, disappointment, modal-payoff focusing, extreme-outcome screening, and limited attention. Applied to three lottery-choice datasets, the data exhibit detectable parsimony: for a majority of subjects, observed concentration exceeds what standard utility models generate on the same menus. The concentration organizes around salience thinking, modal-payoff focusing, and regret.
\end{abstract}

\medskip

\medskip

\setcounter{page}{1}

\section{Introduction}\label{section:intro}

The value of simple behavioral theory lies in parsimony. Cumulative prospect theory \citep{kahneman_tversky1979,tversky_kahneman1992} organizes the Allais paradox, the reflection effect, and a range of framing effects with a reference-dependent value function and a nonlinear weighting of probabilities. Salience theory \citep{bordalo2012_risk} compresses context-dependent preference reversals, decoy effects, and some violations of transitivity into a single contrast-based attention mechanism. These theories are influential precisely because a small number of primitives organize behavioral patterns that would otherwise require separate explanations.\footnote{The broader behavioral literature proposes many such mechanisms; see, among others,  \citet{KoszegiRabin2006} on reference dependence, \citet{LoomesSugden1986,bell85,Gul1991} on regret and disappointment, \citet{enke2024_attenuation,enke2023_uncertainty} on behavioral attenuation and cognitive uncertainty, and \citet{sims2003,gabaix2014,masatlioglu2012,manzini_mariotti2014} on limited attention and consideration.}

But is this compression real? When a broad set of candidate mechanisms is placed on common footing, how many are actually needed to organize a large risky-choice dataset? An informative answer requires three things. We need a way to measure compressibility that does not depend on estimating any one model in particular. We need a way to validate the resulting measure against alternative explanations - in particular, against the possibility that any apparent compression is what a standard stochastic-utility model would have produced anyway. And we need a way to ensure that the answer is not driven by the specific list of candidate mechanisms one happens to entertain. If behavior turns out to be genuinely compressible by these criteria, the enterprise of simple mechanism-based modeling in risky choice is empirically vindicated; if not, the apparent success of such theories may reflect convenient modeling vocabulary, or experimental designs constructed to elicit the very behaviors the theories predict, rather than low-dimensional behavioral structure in choice.

This paper develops the three pieces and applies them to several lottery-choice datasets. We define a Maximum Rule Concentration Index (MRCI) over a finite library of transparent, parameter-free decision rules drawn from canonical behavioral theories of risky choice: salience, regret, disappointment, modal-payoff focusing, extreme-outcome screening, and limited attention. Each rule can be stated in a single sentence and recommends one option whenever a simple feature of the menu makes that option preferable. The MRCI measures the largest concentration of rule usage attainable across menu-by-menu admissible assignments, where admissibility requires that each observed choice be locally rationalizable by the rule assigned to it. When the MRCI is close to one, a single rule can account for nearly all menus; when it is low, any admissible rationalization must spread weight across many rules. The MRCI thus operationalizes parsimony as a quantitative compressibility statistic conditional on a maintained set of candidate rules.

Two natural concerns about this setup organize the rest of the paper. The first is that a high observed MRCI need not, on its own, indicate a meaningful behavioral phenomenon. Under any reasonable stochastic-choice model, some apparent rule concentration arises mechanically, because the model's choice probabilities happen to align with the recommendations of one or another rule at one or another menu. To address this, we develop a benchmark test that treats the MRCI as a property of the data and asks whether conventional stochastic-utility models, fitted to the same data, reproduce that property. 

On the inferential side, the test has exact finite-sample size under the benchmark null with known parameter, conditional on the realized menu sequence. We define an excess-concentration alternative phrased directly in terms of the MRCI - requiring that the realized MRCI exceed, in the limit, a frontier determined by the null model's expected per-rule admissibility - and prove that, under this alternative, the exact benchmark $p$-value vanishes in probability and the test rejects with probability tending to one at every nominal level. 

The second concern is that the MRCI is by construction conditional on the maintained library: a critic could worry that any few-rule result is an artifact of the rules a researcher happened to include. To address this, we define a reduced-library frontier that selects the best small sublibrary of rules on a training sample of subjects and evaluates retention of compression out of sample on held-out subjects. Because rule selection uses only internal concentration criteria, the exercise is a genuine internal robustness check rather than a second specification search.

We apply the framework to three datasets: the CPC18 lottery-choice experiment of \citet{erev_ert_plonsky_2017} and \citet{plonsky2017_cpc18}, with $686$ subjects making between $400$ and $560$ binary lottery decisions each (the feedback portion of the experiment, see Section~\ref{sec:empirical}); the predecessor CPC15 dataset \citep{erev_ert_plonsky_2017}, with $279$ subjects on a different problem set; and the incentivized one-shot environment of \citet{baillon2020}---hereafter BBS---with $139$ subjects each making $70$ choices over richer multi-outcome lotteries.

The empirical pattern is the same in structure across all three datasets but differs sharply in strength. The average raw MRCI is high in each ($0.545$ in CPC18, $0.525$ in CPC15, $0.597$ in BBS), but most of this is mechanical: a pair of attention rules built into the library to ensure the framework is always feasible generates an average floor near $0.52$. The substantive object is the excess above this floor. CPC18 carries the strongest signal: mean excess of $0.028$ on a $[0,1]$ scale, $90$th percentile $0.085$, a $95\%$ subject-level bootstrap CI strictly above zero, and a majority of subjects rejecting each of the fitted stochastic-utility benchmarks. BBS comes next (mean excess $0.015$), with $14\%$ of subjects rejecting the cumulative prospect theory (CPT) benchmark at the $5\%$ level and the active-rule composition shifting toward reference-dependent mechanisms consistent with the experimental design. The signal is weakest on CPC15 (mean excess $0.006$, an order of magnitude smaller than CPC18); the framework remains portable but the substantive few-rule pattern is at the threshold of detectability.

A more demanding benchmark sharpens the picture. We also test observed compression against BEAST \citep{erev_ert_plonsky_2017,plonsky_etal2025_nhb}, the leading predictive model of risky choice in the CPC15/CPC18 problem space. Implemented from the official simulator without modification, BEAST absorbs part of the few-rule signal: CPT's $61.5\%$ rejection rate at the $5\%$ level falls to $52.9\%$ under BEAST. A near-majority of subjects still display compression exceeding even BEAST. 


The reduced-library frontier formalizes how few rules suffice. The two attention rules are retained throughout for feasibility; the question is which of the six remaining rules - worst-case focusing, best-case focusing, modal-payoff focusing, salience, regret, and disappointment - are needed to recover the compression. For each sublibrary size from one to six, we select, on a $50\%$ training sample of subjects, the sublibrary of that size that maximizes cross-subject mean MRCI; we then evaluate \emph{excess retention}, the share of the full-library excess MRCI that the sublibrary recovers above the mechanical floor, on the held-out half, and repeat over $100$ random splits to assess stability. The compression concentrates on a small core, and selection is highly stable. Salience alone, selected in $100\%$ of splits as the best one-rule sublibrary, preserves about $40\%$ of the excess; modal-payoff focusing plus salience, selected in $95\%$ of splits as the best pair, preserves $71\%$; adding regret as a third rule reaches $82\%$ ($80\%$ of splits). Four rules - salience, modal-payoff focusing, regret, and best-case focusing - reach $90\%$ excess retention; disappointment is never selected at three rules or fewer.

On BBS we also run a behavioral benchmark grounded in reference-dependent choice: the multi-reference specification of \citet{baillon2020}, which allows for six alternative reference points and identifies, per subject, the one best supported by the data. Only $1.4\%$ of BBS subjects reject this benchmark at the $5\%$ level, far below the rate that reject the generic stochastic-utility benchmarks on the same dataset---a reference-dependent model with subject-specific reference points absorbs much of what our rule library captures. On CPC18, by contrast, a near-majority of subjects still reject even BEAST; the contrast is what makes the CPC18 few-rule structure substantively informative.

The headline survives a battery of robustness checks. The MRCI is invariant to within-problem dynamics by construction: because all repetitions of a given problem face the same menu, only the count of left versus right choices on each problem enters the statistic, so learning and autocorrelation across repetitions may not drive the result. The standard benchmark holds each subject's realized sample of games fixed and randomizes only choices, so the rejection rate is conditional on that game sample; to assess sensitivity to game composition, a problem-level block bootstrap that also resamples games yields a $23\%$ CPT rejection rate at the $5\%$ level (against $62\%$ under iid CPT), bracketing a methodological range. The formal size guarantee assumes known benchmark parameters, but the application estimates them subject by subject; a Monte Carlo size calibration confirms that the resulting parametric bootstrap may be conservative under the CPT null, so rejection rates may not be inflated by parameter estimation. Random left-right relabeling collapses the MRCI to near the mechanical floor, confirming that the compression reflects the actual choice patterns rather than a structural feature of the menus.

The paper contributes to a broad question in economics: whether the empirical success of simple behavioral theories reflects genuine low-dimensional structure in behavior, or merely convenient modeling vocabularies for a higher-dimensional process. The literature on parsimony in behavior has approached this question along two lines. A first line proposes parsimonious models, often with microfoundations, that organize a range of choice phenomena within a single mechanism family \citep{kahneman_tversky1979,tversky_kahneman1992,bordalo2012_risk,KoszegiRabin2006}. A second line treats parsimony as an empirical question about the joint structure of behavior across choice environments  (\citep{chapman_etal2023_econographics, enke2024_attenuation}). The present paper develops a third line: parsimony is measured directly on choice data by quantifying how compressible a subject's choices are by a library of decision rules. Within the literature on procedural choice and decision rules - including bounded rationality and heuristic choice \citep{simon1955,rubinstein88,tversky_kahneman1974,Kahneman1982}, rule rationality \citep{aumann2008,aumann2019}, simplicity in risky choice \citep{enke2023,fudenberg2022,puri2025,oprea2024,hu2023}, the priority heuristic and lexicographic models \citep{brandstatter_et_al2006,gigerenzer_gaissmaier2011}, and procedural models of consideration and categorization \citep{salant2006,manzini2007,manzini2012,masatlioglu2012,caplin_dean2015,cherepano2013,arrieta2024,halevy2024} - we provide a quantitative measure of the compressibility of behavior into a small number of mechanism families together with a formal benchmark test.

The closest companion paper is \citet{seror2026randomrulemodel}, with which we share the rule library and the activation discipline. The two papers ask different questions and use different methods; their findings do not depend on each other. \citet{seror2026randomrulemodel} develops a Random Rule Model: an explicit data-generating process in which the rule applied at each choice is drawn from a menu-feature softmax gate over rule indices. The contribution is structural and predictive. The MRCI is a non-parametric measurement of the parsimony within a dataset; the Random Rule Model is a candidate structural account of how that parsimony might arise. 

The remainder of the paper is organized as follows. Section~\ref{sec:model} introduces the choice environment, the rule library, and the activation discipline. Section~\ref{sec:MRCI} defines the MRCI, the reduced-library frontier, and their key properties. Section~\ref{sec:inference} develops the benchmark test against stochastic-utility models. Section~\ref{sec:empirical} presents the empirical results for CPC18 - raw and excess MRCI, benchmark validation, library reduction, robustness followed by replication on CPC15 and cross-environment portability on BBS. Section~\ref{sec:conclusion} concludes. Appendices contain formal rule definitions, all proofs, and computational details.

\section{Framework}\label{sec:model}

To measure the compressibility of risky-choice behavior, we need three ingredients: a library of candidate mechanisms, a discipline for deciding which mechanisms can rationalize which choices, and an object that captures the strongest possible compression of the data into those mechanisms. This section introduces the first two ingredients; Section~\ref{sec:MRCI} defines the compression measure itself and Section~\ref{sec:inference} develops the benchmark against stochastic utility. We develop the framework for binary lottery choices, which is the setting of our empirical application. The extension to discrete choice over menus of arbitrary size is conceptually straightforward and mainly notational.

\paragraph{Choice environment.}
We consider a decision maker facing a sequence of binary choices under risk. The dataset is a collection of observations $D=\{(x_t,A_t)\}_{t\in\mathcal T}$ indexed by $\mathcal T:=\{1,\dots,T\}$. At each observation, the decision maker faces a menu $A_t=\{L_t^1,L_t^2\}$ consisting of two finite-support lotteries over real monetary outcomes, and selects one alternative $x_t\in A_t$. We write $L_t^-$ for the unchosen alternative.

\subsection{Decision rules}\label{subsec:rules}

We model behavior as the outcome of a small library of \emph{decision rules}, in the spirit of \citet{seror2026randomrulemodel}. Each rule is a deterministic mapping from a menu to a recommendation or to abstention - it either picks one of the two options or declines to rank them. Formally, for each rule $f$ and menu $A=\{L^1,L^2\}$, define an \emph{activity indicator}
\[
A_f(A)\ :=\ \begin{cases} 1 & \text{if } f \text{ delivers a strict recommendation at } A, \\ 0 & \text{otherwise (}f\text{ abstains)}, \end{cases}
\]
and, when $A_f(A)=1$, a \emph{recommendation indicator} $L_f(A)\in\{L^1,L^2\}$ identifying the option recommended by $f$. Rules are parameter-free: their recommendations depend only on the menu, not on any estimated quantity. The pair $(A_f, L_f)$ is the only object the framework requires from each rule, and the abstract setup does not specify how rules are constructed internally---any procedure that delivers a deterministic, parameter-free mapping from menus to $\{L^1, L^2, \text{abstain}\}$ qualifies.

In our application, we construct a concrete library by specifying, for each rule, a \emph{perceived-payoff transformation} that maps the original menu to a single representative payoff per option---the worst-case outcome for one rule, the modal outcome for another, the salience-weighted contrast for a third, and so on. The activity and recommendation indicators are then derived by a strict comparison: $f$ recommends the option with the strictly higher representative payoff and abstains in case of a tie. This two-step construction is an implementation device that gives a uniform machinery for translating diverse heuristics into the indicator form; it is not part of the abstract framework. Full formal definitions of the transformations for each rule are in Appendix~\ref{ec:rules}.

The maintained baseline library contains eight rules, each capturing a behavioral mechanism that selects an option on the basis of a particular menu feature:
\begin{itemize}[leftmargin=*,itemsep=0.15em]
\item Three outcome-simplification heuristics, each reducing a lottery to a single representative payoff: worst-case (\textsc{MMn}), best-case (\textsc{MMx}), and most-likely payoff (\textsc{MAP}). These rules formalize the focal-feature heuristics studied in the bounded-rationality and lexicographic-choice traditions \citep{simon1955,gigerenzer_gaissmaier2011,brandstatter_et_al2006}.
\item Three context-dependent comparison rules: salience (\textsc{SAL}, \citealp{bordalo2012_risk}), which uses the payoff comparison with the largest contrast; regret (\textsc{REG}, \citealp{LoomesSugden1986}), which minimizes the worst foregone gain; and disappointment (\textsc{DIS}, \citealp{bell85,Gul1991}), which avoids the option with the largest downside relative to its most likely outcome.
\item A pair of \emph{limited-attention rules}, \textsc{A1} and \textsc{A2}, capturing one-sided consideration of the menu. \textsc{A1} corresponds to a decision maker who attends only to the left option and picks it; \textsc{A2}, symmetrically, attends only to the right option and picks it. This is the limited-consideration channel familiar from the literature on attention and consideration sets \citep{manzini2007,masatlioglu2012,manzini_mariotti2014}: a decision maker who, faced with a binary menu, effectively considers only one side. \textsc{A1} and \textsc{A2} are substantive behavioral rules in their own right; they also account for a large mechanical share of the raw MRCI, with the rest of the analysis turning on the gap between raw MRCI and the floor that one-sided attention alone would generate.
\end{itemize}

Each rule in the library is a parameter-free implementation of a comparison primitive familiar from a richer behavioral theory: \textsc{SAL}, \textsc{REG}, and \textsc{DIS} embed the comparison logic of salience theory \citep{bordalo2012_risk}, regret theory \citep{LoomesSugden1986}, and disappointment theory \citep{bell85,Gul1991}; \textsc{MMn}, \textsc{MMx}, and \textsc{MAP} formalize the focal-feature heuristics studied in the bounded-rationality tradition \citep{simon1955,gigerenzer_gaissmaier2011,brandstatter_et_al2006}; and \textsc{A1}, \textsc{A2} embed the limited-consideration channel familiar from the attention literature \citep{manzini2007,masatlioglu2012}. The richer theories combine these comparisons with parametric valuation; our rules only retain simplified comparison primitives. This is a methodological choice: it ensures that each rule has clear, falsifiable behavioral content and that the library can be tested as a unit. A finding that this parsimonious version of the literature already organizes observed behavior into a small number of mechanism families is, if anything, stronger than the analogous claim drawn from richer parametric specifications, since it does not rely on the additional flexibility that the specifications may introduce.

The library is deliberately coarse and parameter-free in a further sense: each rule can be stated in a single sentence and contains no estimated curvature, weighting, or threshold parameter fit to the data. Probability-weighting rules (Prelec, rank-dependent) are notably absent because they cannot be expressed without an estimated curvature parameter; their inclusion would blur the distinction between transparent decision rules and parametric utility specifications. Maintaining that distinction is what gives the rule library its empirical content.

All objects in this paper are conditional on the maintained library. Because enlarging the library can only weakly increase attainable concentration (Proposition~\ref{prop:monotone-necessary}), a coarse baseline yields a conservative measure of parsimony. Robustness exercises with alternative family partitions and richer libraries are reported in Section~\ref{subsec:robustness}.

\subsection{Activity and admissible rule assignments}\label{subsec:activation}

Not every rule delivers a clear-cut recommendation at every menu, and even when a rule does deliver a recommendation, that recommendation need not match the realized choice. The framework restricts attention to rules that are decisive and consistent with the data: following \citet{seror2026randomrulemodel}, we say that rule $f$ is \emph{active} at observation $t$ if it delivers a strict recommendation at $A_t$ and that recommendation coincides with the realized choice $x_t$. In indicator form,
\[
f\in\mathcal F^{\mathrm{act}}_t \quad\Longleftrightarrow\quad A_f(A_t)\ =\ 1 \quad\text{and}\quad L_f(A_t)\ =\ x_t.
\]
In our concrete implementation, this criterion is easy to state. Each rule reduces the two lotteries to a single representative payoff---the worst-case outcome for \textsc{MMn}, the modal payoff for \textsc{MAP}, the salience-weighted contrast for \textsc{SAL}, and so on. Rule $f$ is then active at observation $t$ when the chosen lottery's representative payoff is strictly greater than the unchosen lottery's: the rule says ``pick the option with the strictly higher value of this feature,'' and the realized choice complies. Formal definitions of the perceived-payoff transformations for each rule are given in Appendix~\ref{ec:rules}.\footnote{The strict-recommendation criterion adopted here is the same as in \citet{seror2026randomrulemodel}, who reports that it yields the lowest out-of-sample prediction error relative to alternative activation thresholds.}

The active set $\mathcal F^{\mathrm{act}}_t$ collects the rules that, applied to menu $t$, would have produced the observed choice as their decisive recommendation.

Because \textsc{A1} and \textsc{A2} always recommend the left and right option respectively, exactly one of the two attention rules is always in $\mathcal F^{\mathrm{act}}_t$, so the active set is never empty.

\begin{definition}[Admissible rule assignment]\label{def:admissible}
An assignment $\mathbf y=\{y_{t,f}\}_{t\in\mathcal T,f\in\mathcal F}$ with $y_{t,f}\in\{0,1\}$ is admissible if, for each $t$, exactly one rule is used and that rule is active:
\[
\sum_{f\in\mathcal F} y_{t,f}\ =\ 1, \qquad y_{t,f}=1\ \Longrightarrow\ f\in\mathcal F^{\mathrm{act}}_t.
\]
Let $\mathcal C(D;\mathcal F)$ denote the set of admissible assignments.
\end{definition}

\paragraph{Beyond risky choice.} The framework is structurally domain-agnostic. The active set $\mathcal F^{\mathrm{act}}_t$ and Definition~\ref{def:admissible} depend only on (i) a library of deterministic, parameter-free menu-to-recommendation maps and (ii) the realized choice; nothing in the construction is specific to lotteries. The substantive content of the framework---what ``compressibility'' picks out in a given domain---is delivered by the rule library, which is necessarily domain-specific. Once a library is fixed, the MRCI of Section~\ref{sec:MRCI} and the benchmark test of Section~\ref{sec:inference} apply without modification, and what is empirically informative---the magnitude of the excess MRCI above the attention-rule mechanical floor---depends on the design of the choice environment and on the alignment between subjects' choices and the candidate library, just as in our risky-choice application.

\section{Measuring Compression}\label{sec:MRCI}

Given a library of candidate mechanisms and a discipline for assigning them to menus, the central question is: how compressible is the data? Can behavior be organized as if a single mechanism---or a small number of them---were repeatedly invoked across many menus? This section defines the object that measures this compression.

Using the admissible assignments defined in Section~\ref{subsec:activation}, we write $\mathcal C(D;\mathcal F)$ for the set of admissible assignments and $\mathcal C(D)$ as shorthand.

\subsection{The Maximum Rule Concentration Index}

Given an admissible assignment $\mathbf y\in\mathcal C(D)$, define the frequency with which rule $f$ is used as
\[
s_f(\mathbf y):=\frac{1}{T}\sum_{t \in \mathcal T} y_{t,f},
\qquad f\in\mathcal F,
\]
and the associated Herfindahl-Hirschman index (HHI) by
\[
\mathrm{HHI}(\mathbf y):=\sum_{f\in\mathcal F} s_f(\mathbf y)^2\ \in\ \Big[\tfrac{1}{\mid \mathcal F \mid },1\Big].
\]
$\mathrm{HHI}(\mathbf y)$ equals $1/\mid \mathcal F \mid$ when all rules are used equally often, and equals $1$ when a single rule accounts for all observations in $\mathcal T$. 

\begin{definition}[Maximum Rule Concentration Index]\label{def:MRCI}
The Maximum Rule Concentration Index is
\[
\MRCI(D;\mathcal F)
\ :=\ 
\max_{\mathbf y \in \mathcal C(D;\mathcal F)} \mathrm{HHI}(\mathbf y),
\]
and the effective number of rules is $N_{\mathrm{eff}}:=1/\MRCI(D;\mathcal F)$.
\end{definition}

The MRCI measures how parsimoniously observed behavior can be organized within the maintained library $\mathcal F$. An MRCI close to one means the data admit an admissible explanation in which nearly all menus are assigned to the same rule - a single mechanism unifies the dataset. A low MRCI means any admissible explanation must spread weight across many mechanisms. The reciprocal $N_{\mathrm{eff}}=1/\MRCI$ gives the effective number of active rules within $\mathcal F$. Both the MRCI and $N_{\mathrm{eff}}$ are conditional on the maintained library: they measure revealed parsimonious organization, not the true number of latent psychological mechanisms in any structural sense. Enlarging the library can only weakly increase the MRCI (Proposition~\ref{prop:monotone-necessary} below), so a coarse baseline library yields a conservative lower bound on attainable parsimony. Let
\[
\alpha\ :=\ \frac{1}{T}\sum_{t\in\mathcal T}\mathbf 1\{x_t=L_t^1\}
\]
denote the share of observations in which the decision maker chooses the first-listed lottery in the menu.

\paragraph{Basic properties.}
\begin{proposition}\label{prop:properties}
For any dataset $D$, and library $\mathcal F$:
\begin{enumerate}\itemsep0.25em
\item $\MRCI(D;\mathcal F)\le 1$, and if $\{\textsc{A1},\textsc{A2}\}\subseteq\mathcal F$ then
\[
\MRCI(D;\mathcal F)\in \big[\alpha^2+(1-\alpha)^2,\,1\big].
\]
\item $\MRCI(D;\mathcal F)=1$ if and only if there exists a rule $f\in\mathcal F$ such that
\[
f\in\mathcal F_t^{\mathrm{act}}\qquad\text{for all }t\in\mathcal T.
\]
\item Concatenating $k$ identical copies of $D$ (i.e., repeating each observation $k$ times) leaves $\MRCI(D;\mathcal F)$ unchanged.
\item If for some $t\in\mathcal T$ there is a unique active rule $f$ (i.e., $\mathcal F_t^{\mathrm{act}}=\{f\}$), then every maximizer $\mathbf y^\star$ of $\MRCI(D;\mathcal F)$ satisfies $y^\star_{t,f}=1$.
\end{enumerate}
\end{proposition}

The upper bound in (i) corresponds to a one-rule explanation. The lower bound comes from the always-feasible attention-rule assignment: use \textsc{A1} whenever $x_t=L_t^1$ and \textsc{A2} otherwise, which yields $\mathrm{HHI}=\alpha^2+(1-\alpha)^2$. Property (iii) shows that MRCI is intensive (composition-based) rather than extensive (sample-size-based). Property (iv) highlights ``forced moves'' that anchor the optimization: if an observation can only be strictly discriminated by a single rule, every admissible maximizer must assign that observation to the same rule.

\subsubsection*{Monotonicity}

Let $\mathcal F'\supseteq\mathcal F$ be an expanded library.

\begin{proposition}[Monotonicity]\label{prop:monotone-necessary}
For any dataset $D$,
\[
\MRCI(D;\mathcal F')\ \ge\ \MRCI(D;\mathcal F).
\]
Moreover, a necessary condition for strict inequality is that there exists an MRCI-maximizing assignment $\mathbf y^\star$ under $\mathcal F$ and a new rule $f\in\mathcal F'\setminus\mathcal F$ such that, on a nonempty subset of observations, $f$ is active and can replace (while maintaining admissibility) two distinct rules $r\neq s$ that both have positive shares under $\mathbf y^\star$ (i.e., $s_r(\mathbf y^\star),s_s(\mathbf y^\star)>0$).
\end{proposition}

Intuitively, a new rule increases concentration only if it can act as a ``unifier,'' substituting for multiple distinct positive-share rules on the observations where those rules were previously needed.

\label{para:per_problem}In the empirical applications below, subjects face a small number of distinct lottery problems repeated many times (twenty-five times each in CPC18). It is therefore important to be precise about what the MRCI extracts from such data. Because all repetitions of a given problem face the same menu, the active set $\mathcal F_t^{\mathrm{act}}$ is constant across those repetitions. The contribution of a problem to the MRCI optimization therefore depends only on the \emph{count} of $L^1$ versus $L^2$ choices the subject made on that problem, not on the order or trial-level dynamics of those choices. As a consequence, two subjects who chose $L^1$ on the same fraction of repetitions of every problem produce the same MRCI, regardless of any differences in within-problem learning or autocorrelation. The substantive content of the MRCI in repeated-menu data is the alignment between subjects' \emph{per-problem choice frequencies} and the rule prescriptions, aggregated across the small number of distinct problems. Section~\ref{subsec:robustness} verifies this invariance directly: simulating from each subject's empirical per-problem $L^1$-frequencies leaves the MRCI essentially unchanged. 

\paragraph{Computation.}
Computing $\MRCI(D;\mathcal F)$ requires maximizing a convex function (the HHI) over a combinatorial feasible set of binary assignments---a class of problems that is NP-hard in general.\footnote{The problem is a special case of maximizing a separable convex objective over a transportation polytope with binary constraints. } We solve it with a scalable greedy heuristic with random restarts, exploiting the fact that concentration-maximizing assignments load heavily on high-coverage rules. We also provide an exact Mixed Integer Quadratic Programming (MIQP) algorithm to compute the MRCI, although this algorithm does not scale easily. Benchmarking the heuristic against the exact MIQP on 3,430 subsampled instances yields a mean absolute error of our heuristic below $10^{-3}$ and more than $96\%$ of estimates within a $1\%$ tolerance. Details are in Appendices~\ref{sec:computational} and~\ref{appendix:miqp}.

\section{Validating Compression Against Stochastic Utility}\label{sec:inference}

The MRCI is a property of the realized choice data: it summarizes how parsimoniously the observed sequence of choices can be organized in the maintained library. The question this section addresses is structural: do conventional stochastic-utility models, fitted to the same data, reproduce that property? Some concentration arises mechanically under any stochastic-choice model, because the model's choice probabilities happen to align with certain rules at certain menus, and a model can match a subject's per-menu choice probabilities without matching the joint cross-trial pattern that the MRCI captures. We develop a benchmark test that compares observed MRCI to the MRCI distribution generated by a fitted utility model with i.i.d.\ logit noise. Rejection means the joint specification of utility plus i.i.d.\ noise produces less rule-aligned cross-trial structure than the data exhibit; it does not say which component of that joint specification is responsible. Non-rejection, conversely, says the fitted specification is rich enough to reproduce the observed compression. Read this way, the test calibrates how much of the data's parsimony each utility model accounts for.

The construction is intuitive. Fix a subject. We estimate a parametric stochastic-choice model on that subject's data, simulate $B$ synthetic datasets in which each binary choice is drawn independently from the fitted model on the subject's actual menu sequence, recompute the MRCI on each synthetic dataset, and compare the observed MRCI to the resulting Monte Carlo distribution. If the observed value lies above the upper tail of the simulated distribution, the data exhibit compression that the fitted model with i.i.d.\ noise does not generate - structure in the joint cross-trial pattern of rule-alignment that the model's marginal choice probabilities, sampled independently, fail to reproduce. We proceed in four steps: (i) the benchmark null and the Monte Carlo $p$-value, with a finite-sample size guarantee under known parameters; (ii) the special case of menu-independent randomness, which reduces to a permutation test; (iii) an excess-concentration alternative phrased directly in terms of the MRCI, against which we prove the test consistent, with a Random Rule Model used as interpretive language following \citet{seror2026randomrulemodel}; and (iv) three canonical benchmark specifications. All proofs are collected in Appendix~\ref{ec:proofs}.

\subsection{Benchmark null and the Monte Carlo test}\label{subsec:model_benchmark}

Fix a subject and, for each sample size $T$, let $A_T = \{A_{t,T}\}_{t=1}^T$ denote the deterministic sequence of binary menus the subject faced, with $A_{t,T} = \{L^1_{t,T}, L^2_{t,T}\}$. Write $d_{t,T} := \mathbf 1\{x_{t,T} = L^1_{t,T}\}$ for the left-choice indicator and $M(D_T) := \MRCI(D_T;\mathcal F)$ for the MRCI of the dataset.

The benchmark null asserts that, conditional on $A_T$, the indicators $(d_{t,T})_{t=1}^T$ are independent Bernoulli random variables with menu-specific probabilities $q_{t,T}(\theta_0) \in [0,1]$ generated by some parametric stochastic-utility model with parameter $\theta_0\in\Theta$.

\begin{assumption}[Utility benchmark null ($H_0$)]\label{ass:u0}
There exists a parameter $\theta_0\in\Theta$ such that, conditional on the deterministic menu sequence $A_T$, the random variables $(d_{t,T})_{t=1}^T$ are independent and satisfy
\[
\Pr_{\theta_0}(d_{t,T}=1\mid A_T)=q_{t,T}(\theta_0)\in[0,1],\qquad t=1,\dots,T.
\]
\end{assumption}

A draw from this null is a synthetic dataset $D_T^* = \{(x_{t,T}^*, A_{t,T})\}_{t=1}^T$ on the same menu sequence in which $(d_{t,T}^*)_{t=1}^T$ are independent Bernoulli random variables with success probabilities $q_{t,T}(\theta_0)$. Given $B\geq 1$ independent draws $D_T^{*(1)},\dots,D_T^{*(B)}$, we use the standard right-tail Monte Carlo $p$-value:
\begin{equation}\label{eq:pu_mc}
\hat p^U_{T,B}(D_T;\theta_0)
:=
\frac{1+\sum_{b=1}^B \mathbf 1\{M(D_T^{*(b)})\geq M(D_T)\}}{B+1}.
\end{equation}

When the benchmark parameter $\theta_0$ is known (the oracle case), this Monte Carlo test has exact finite-sample size:

\begin{proposition}[Finite-sample validity under known benchmark parameter]\label{prop:utility_exact}
Under Assumption~\ref{ass:u0}, for every $T\geq 1$, every $B\geq 1$, and every $\eta\in[0,1]$,
\[
\Pr_{\theta_0}\!\left(\hat p^U_{T,B}(D_T;\theta_0)\leq \eta \,\middle|\, A_T\right)\leq \eta.
\]
Hence the test that rejects when $\hat p^U_{T,B}(D_T;\theta_0)\leq \eta$ has exact finite-sample size at most $\eta$ conditional on the realized menu sequence.
\end{proposition}

The proof is given in Appendix~\ref{ec:proofs} and uses the exchangeability of the observed and simulated statistics conditional on $A_T$. Proposition~\ref{prop:utility_exact} is an oracle result: it assumes the benchmark parameter $\theta_0$ is known. In the empirical application, $\theta_0$ is estimated subject-by-subject by maximum likelihood and the benchmark is simulated from $q_{t,T}(\hat\theta)$, yielding a parametric bootstrap. Exact finite-sample validity is then lost, but the test remains a well-calibrated benchmark for evaluating observed concentration; we assess the size of this approximation directly in the empirical analysis (Section~\ref{subsec:robustness}).

\subsection{The permutation test as a special case}\label{subsec:permutation_special}

If $q_{t,T}(\theta_0) \equiv \alpha$ for all $t$---i.e., the benchmark imposes the same left-choice probability at every menu---the indicators are i.i.d.\ $\mathrm{Bernoulli}(\alpha)$, and the conditional distribution of $(d_{t,T})$ given $Z = \sum_t d_{t,T}$ is uniform over all binary vectors with exactly $Z$ ones. The Monte Carlo test \eqref{eq:pu_mc} then reduces to a permutation test: compute MRCI for $B$ random reassignments of the $Z$ left-choices across $T$ positions, holding the menu sequence fixed, and compare the observed MRCI to this conditional randomization distribution. This nonparametric special case requires no model estimation and serves as a minimal screen for non-random structure, in the spirit of permutation-based revealed-preference inference \citep[see, e.g.,][]{cherchye2025}. About one-third of CPC18 subjects fail to reject even this weak benchmark, suggesting that their observed concentration is mechanically attainable under menu-independent randomness; the remainder display structure that exceeds what i.i.d.\ left-choice randomness would produce.

\subsection{Excess-concentration alternative and consistency}\label{subsec:rrm_consistency}

To formulate the alternative against which we prove the test consistent, we phrase a hypothesis directly in terms of the MRCI. Activation discipline enters the construction of the admissible set $\mathcal C(D;\mathcal F)$ (Section~\ref{subsec:activation}) and, through it, of the MRCI itself; it does not enter the alternative as a constraint on the data-generating process. We use a simple Random Rule Model as a verbal device for interpreting what ``excess concentration'' means, but the formal alternative does not rely on it.

\begin{definition}[Random Rule Model]\label{def:RRM}
Fix a library $\mathcal F$. A Random Rule Model (RRM) is a probability vector $w=(w_f)_{f\in\mathcal F}\in\Delta^{|\mathcal F|-1}$ with associated coincidence probability
\[
\tau(w)\;:=\;\sum_{f\in\mathcal F}w_f^2,
\]
the probability that two independent draws from $w$ coincide.
\end{definition}

Higher $\tau(w)$ corresponds to a more concentrated rule-share vector; an MRCI close to $\tau(w)$ for some $w$ with high $\tau(w)$ is what we mean informally by ``few-rule'' organization. The key inferential question is whether the observed MRCI exceeds a benchmark level $\kappa$ that the null model can mechanically generate. We formalize the alternative as a direct lower bound on the observed MRCI.

\begin{assumption}[Excess-concentration alternative ($H_1$)]\label{ass:H1}
Let $P_T$ denote the true law on $D_T$ conditional on $A_T$. There exists $\bar\tau>\kappa$ such that
\[
P_T\!\left(M(D_T)\geq \bar\tau\right)\to 1
\qquad\text{as }T\to\infty.
\]
\end{assumption}

The benchmark frontier $\kappa$ depends on the null model. For the model-based benchmark, let
\[
c_{f,T}(\theta_0)\ :=\ \E_{\theta_0}\!\Big[T^{-1}\sum_{t=1}^T \mathbf 1\{f\in\mathcal F_t^{\mathrm{act}}(D_T^*)\}\,\Big|\, A_T\Big]
\]
denote the expected share of menus at which rule $f$ is active under the benchmark, and define the \emph{utility benchmark frontier}
\begin{equation}\label{eq:kappa_u}
\kappa_U(\theta_0)\ :=\ K(c(\theta_0)),
\qquad
K(c)\ :=\ \max\Big\{\sum_f s_f^2 : s_f\geq 0,\ \sum_f s_f=1,\ s_f\leq c_f\Big\}.
\end{equation}
The function $K(\cdot)$ is the maximal HHI achievable when each rule's share is capped at $c_f$; its maximizer is given by water-filling on the sorted shares. For the permutation special case (Section~\ref{subsec:permutation_special}), the same construction with menu-independent randomness yields $\kappa = \tau_0(\alpha):=\alpha^2+(1-\alpha)^2$, the concentration obtained by mixing only \textsc{A1} and \textsc{A2} with shares $(\alpha,1-\alpha)$.

\begin{proposition}[Consistency against excess concentration]\label{prop:utility_power}
Fix the deterministic menu sequence $A_T$. Let $Q_{T,\theta_0}$ denote the benchmark law on $D_T$ satisfying Assumption~\ref{ass:u0} with parameter $\theta_0$, and suppose the benchmark expected admissibility shares $c_{f,T}(\theta_0)$ converge to limits $c_f(\theta_0)$ for every $f\in\mathcal F$. Define $\kappa_U(\theta_0):=K(c(\theta_0))$.
Let $P_T$ denote the true data-generating law for $D_T$ and suppose there exists $\bar\tau>\kappa_U(\theta_0)$ such that
\[
P_T\!\left(M(D_T)\geq \bar\tau\right)\to 1\qquad\text{as }T\to\infty.
\]
Then the exact benchmark $p$-value
\[
p_T^U(D_T;\theta_0)\ :=\ Q_{T,\theta_0}\!\left(M(D_T^*)\geq M(D_T)\,\big|\,A_T\right)
\]
satisfies $p_T^U(D_T;\theta_0)\xrightarrow{p}0$ under $P_T$, and for every fixed $\eta\in(0,1)$, $P_T(p_T^U(D_T;\theta_0)\leq \eta)\to 1$. The same conclusion holds for the Monte Carlo $p$-value whenever $B_T\to\infty$.
\end{proposition}

The proof, given in Appendix~\ref{ec:proofs}, proceeds in two steps: (i) under $P_T$, $M(D_T)\geq \bar\tau$ in probability; (ii) under $Q_{T,\theta_0}$, benchmark draws satisfy $M(D_T^*)\leq \kappa_U(\theta_0)$ in probability, because expected admissibility frequencies cap each rule's share. When $\bar\tau>\kappa_U(\theta_0)$, the two bounds separate and the $p$-value vanishes. The construction puts no requirement on how the alternative arises---rule-based or otherwise---only that the realized MRCI cross the benchmark frontier $\kappa_U(\theta_0)$ in the limit. An RRM with weights $\pi$ such that the data admit, in the limit, an admissible assignment of average HHI $\tau(\pi)>\kappa_U(\theta_0)$ is one such alternative; nothing in the proof requires the latent labels themselves to be admissible, and the excess-concentration condition on $P_T$ is the only restriction imposed on the data process.

\subsection{Three benchmark specifications}\label{subsec:specifications}

We implement the test with three canonical stochastic-choice models, ordered from simplest to most flexible. Each evaluates lotteries through a parametric functional and chooses stochastically via a logistic error rule, yielding menu-dependent left-choice probabilities $q_t(\theta)$ that plug directly into the Monte Carlo construction of Section~\ref{subsec:model_benchmark}.

\paragraph{Expected Value (EV).}
The risk-neutral benchmark. Let $\mathrm{EV}(L):=\sum_k z_k\pi_k$ denote the expected value
of lottery $L$. The left-choice probability is
\[
q_t^{\mathrm{EV}}(\lambda)
\;=\;
\Lambda\!\left(\lambda\left[\mathrm{EV}(L_t^1)-\mathrm{EV}(L_t^2)\right]\right),
\]
where $\Lambda(x)=(1+e^{-x})^{-1}$ is the logistic function and $\lambda>0$ is a precision
(inverse-noise) parameter. The single free parameter $\lambda$ is estimated by maximum likelihood
for each subject.

\paragraph{Expected Utility (EU-CRRA).} Let $u(x;r)=\operatorname{sign}(x)\,\lvert x\rvert^{r}$ denote a sign-symmetric power utility function with curvature $r>0$, and define $\mathrm{EU}(L;r):=\sum_k u(z_k;r)\,\pi_k$.\footnote{On gain-only payoffs $u$ reduces to $x^r$ and is monotone-equivalent (up to an affine rescaling of the noise parameter $\lambda$) to the conventional CRRA form $x^r/r$, so the implied left-choice probabilities are identical. CPC18 contains both gains and losses, and the sign-symmetric extension is used to keep a single free curvature parameter; no separate loss-aversion coefficient is estimated for the EU benchmark.} The left-choice probability is
\[
q_t^{\mathrm{EU}}(r,\lambda)
\;=\;
\Lambda\!\left(\lambda\left[\mathrm{EU}(L_t^1;r)-\mathrm{EU}(L_t^2;r)\right]\right).
\]
The parameter vector $\theta=(r,\lambda)$ is estimated by maximum likelihood for each subject.

\paragraph{Cumulative Prospect Theory (CPT).}
CPT \citep{tversky_kahneman1992} evaluates lotteries using a value function applied to outcomes and a nonlinear weighting of cumulative probabilities. Let $v(x;\alpha)=\operatorname{sign}(x)\,\lvert x\rvert^{\alpha}$ denote the sign-symmetric value function (with curvature $\alpha\in(0,1]$), and let $w(p;\gamma)=\exp(-(-\ln p)^{\gamma})$ denote the Prelec probability weighting function \citep{prelec1998}.\footnote{We adopt the sign-symmetric power form without a separate loss-aversion coefficient. The CPT benchmark is intended as a flexible utility specification that adds probability weighting on top of EU's curvature, and the comparison the benchmark makes does not depend on a separate gain--loss asymmetry; introducing a fourth parameter would only widen the box for benchmark-side fits.} For a lottery $L$ with ordered outcomes $z_1 < \cdots < z_m$ and probabilities $\pi_1,\dots,\pi_m$, the CPT valuation is
\[
\mathrm{CPT}(L;\alpha,\gamma) \;=\; \sum_{k=1}^m v(z_k;\alpha)\left[w\!\left(\sum_{j=k}^m \pi_j;\gamma\right) - w\!\left(\sum_{j=k+1}^m \pi_j;\gamma\right)\right],
\]
where the decision weights are obtained by applying $w$ to cumulative survival probabilities. The left-choice probability takes the Fechner form
\[
q_t^{\mathrm{CPT}}(\alpha,\gamma,\lambda)
\;=\;
\Lambda\!\left(\lambda\left[\mathrm{CPT}(L_t^1;\alpha,\gamma)-\mathrm{CPT}(L_t^2;\alpha,\gamma)\right]\right).
\]
The parameter vector $\theta=(\alpha,\gamma,\lambda)$ is estimated by maximum likelihood for each subject.\footnote{For numerical stability the CPT optimizer reparameterizes $\alpha\in(0,1)$ via a logistic transform and $\gamma\in(0.3,1.5)$ via an affine logistic transform; $\lambda>0$ is unbounded. The $\gamma$ box brackets the empirical mass of published probability-weighting estimates (typically $\gamma\in[0.5,1.0]$, see \citealp{prelec1998,tversky_kahneman1992}) and avoids the two regimes in which the Prelec function is ill-conditioned: as $\gamma\to 0$ the function collapses to a near-constant on $(0,1)$ and the likelihood gradient in $\gamma$ vanishes, and for $\gamma\gtrsim 1.5$ the weighting becomes so extreme S-shaped that small probabilities receive vanishing weight and the likelihood underflows. The boundary diagnostics in Section~\ref{subsec:robustness} use interior cutoffs $\alpha\in(0.05,0.95)$ and $\gamma\in(0.32,1.48)$. For $\gamma$ the cutoffs sit exactly $0.02$ inside the optimization-box edges $(0.3,1.5)$; for $\alpha$ the box is the open interval $(0,1)$ (asymptotic only) and the $(0.05,0.95)$ cutoffs flag subjects whose logistic-reparameterized MLE has saturated at either end. The diagnostic identifies subjects whose unconstrained likelihood ridge runs into the boundary rather than into an interior optimum. The headline rejection rate among interior subjects is essentially unchanged from the full-sample rate, so the box width and the buffer choice are not load-bearing.}

\medskip
For each benchmark model $m\in\{\mathrm{EV},\mathrm{EU},\mathrm{CPT}\}$ and each subject, we
estimate $\hat\theta_m$ by maximum likelihood using the subject's choice data, simulate $B$
synthetic datasets from $q_{t,T}(\hat\theta_m)$ on the same menu sequence, and compute the
Monte Carlo $p$-value as in~\eqref{eq:pu_mc}. The three benchmarks form a hierarchy of
increasing flexibility: EV imposes risk neutrality; EU-CRRA allows for risk aversion but
preserves rational expected-utility structure; CPT adds probability distortion, the dominant
behavioral channel in the literature. Rejection of all three benchmarks provides strong evidence
that the observed rule concentration reflects genuine ``few-rule'' structure that cannot be
attributed to any of these standard models with stochastic noise.

\section{Application to Risky Choice}\label{sec:empirical}

\paragraph{Data.}
We apply our framework to the publicly available CPC18 competition dataset \citep{erev_ert_plonsky_2017,plonsky2017_cpc18}.\footnote{The data can be downloaded at \url{https://cpc-18.com/data/}.} Since our rule library and dominance benchmark are defined over lotteries with known probabilities, we restrict attention to the risk problems and exclude ambiguity tasks. The CPC18 design separates an initial no-feedback block from four subsequent blocks with outcome feedback; the published prediction competitions treat the no-feedback block separately, and we follow that convention by restricting the analysis to the four feedback blocks (20 trials per problem). The resulting sample comprises $N=686$ subjects (average age 24.9, s.d.\ 2.93; 54\% female) recruited across multiple locations in Israel, each facing $20$ to $28$ distinct binary lottery problems for an average of $517$ analyzed decisions (range $400$--$560$). The choice problems were designed to elicit canonical behavioral anomalies in risky choice, including Allais paradoxes \citep{allais1953,tversky_kahneman1986}, certainty-like effects \citep{kahneman_tversky1979}, break-even effects \citep{thaler_johnson1990}, splitting effects \citep{birnbaum2008}, or systematic departures from expected-value reasoning driven by the weighting of rare events \citep{barron_erev2003}. For a detailed description of the experimental design, the exact list of problems, and the mapping from problems to targeted phenomena, see \citet{erev_ert_plonsky_2017}. Summary statistics are provided in Table \ref{tab:summary_stats}.

\paragraph{Implementation.}
For each subject, we compute the MRCI using the greedy heuristic with $100$ random restarts (see Appendix~\ref{sec:computational}). For the model-based benchmark tests, we estimate each specification (EV, EU-CRRA, CPT) by maximum likelihood per subject, simulate $B=500$ synthetic datasets from the fitted model on the same menu sequence, and compute the Monte Carlo $p$-value as in~\eqref{eq:pu_mc}. For the permutation test, we use $B=500$ random permutations. All computations use parallel execution on a 64-core cluster.

Before reporting results, we note that all eight baseline rules have nontrivial activity \emph{coverage} in CPC18 (Figure~\ref{fig:rule_coverage}): each exceeds roughly $15\%$ of menus on average. Low importance for a rule therefore reflects substitutability, not mechanical inadmissibility.

\subsection{How compressible is behavior?}

We organize the empirical analysis around two quantities. The \emph{raw} MRCI summarizes the maximum concentration achievable under the maintained library, including the attention rules \textsc{A1} and \textsc{A2}. The \emph{excess} MRCI subtracts off the subject-specific mechanical floor generated by those attention rules, $\alpha^2 + (1-\alpha)^2$, and is the substantive object of interest: it measures how much additional concentration the data admit beyond what attention rules would mechanically produce.

The average raw MRCI across subjects is $0.545$ (s.d.\ $0.050$), corresponding to an effective number of rules $N_{\mathrm{eff}} \approx 1.85$ (Figure~\ref{fig:neff}). The subject-specific mechanical floor averages $0.517$, so the mean excess MRCI is $0.028$. A nonparametric subject-level bootstrap---resampling the $686$ subjects with replacement and recomputing the cross-subject mean on each of $B=2000$ resamples---gives a $95\%$ confidence interval of $[0.025, 0.031]$ for the population mean excess MRCI (Table~\ref{tab:bootstrap_excess}), strictly above zero. The excess is modest on average but heterogeneous: $72\%$ of subjects have positive excess concentration, $24\%$ have excess above $0.05$, and the 90th percentile reaches $0.085$ (Figure~\ref{fig:excess}). At the same time, no subject in the sample admits a single-rule rationalization: the share of subjects with $N_{\mathrm{eff}}\approx 1$ is $0.000$ (Table~\ref{tab:summary_stats}). Even at its most concentrated, the data require at least two distinct rules to organize hundreds of decisions per subject.

The substantive question is not whether the excess is large in absolute terms but whether it is larger than what standard stochastic-utility models would mechanically generate. The next subsection takes that comparison directly.

\subsection{Compression exceeds what stochastic utility would generate}

If risky-choice data were only compressible under generic noise around standard utility, the case for simple mechanism-based theory would be much weaker. We therefore compare the observed concentration to what conventional stochastic utility models would mechanically generate on the same menus, using the benchmark framework of Section~\ref{subsec:model_benchmark} with the three specifications of Section~\ref{subsec:specifications}. The benchmarks are implemented as parametric bootstraps: for each subject, we estimate the model by maximum likelihood and simulate synthetic datasets from the fitted model. This is calibrated benchmark evidence, not an exact test (exact size is guaranteed only under known parameters; see Proposition~\ref{prop:utility_exact}).

Table~\ref{tab:benchmark_results} and Figure~\ref{fig:test} report the results. Even against cumulative prospect theory with Fechner noise, $52.9\%$ of subjects display concentration exceeding the benchmark at the $1\%$ level. Against expected utility, $57.0\%$ exceed the benchmark; against risk-neutral expected value, $65.3\%$. The permutation test (Section~\ref{subsec:permutation_special}) yields rejection rates of $63\%$ at the $1\%$ level. For a majority of subjects, the observed excess concentration---modest as it is in absolute terms---is larger than what these utility models with noise would produce on the same menus.

Two caveats. First, the CPT benchmark is imprecisely estimated for a nontrivial share of subjects: approximately $38\%$ have estimated value-function curvature at or near the boundary of the parameter space (specifically, $\hat\alpha\to 1$, the linear-utility limit). To verify that this does not drive the headline rejection rate, we condition on subjects with interior CPT estimates ($\hat\alpha\in(0.05, 0.95)$ and $\hat\gamma\in(0.32,1.48)$) and report rejection rates separately (Table~\ref{tab:cpt_boundary}). Among the $55.1\%$ of subjects with both parameters interior, the CPT rejection rate at the $1\%$ level is $50.5\%$---essentially unchanged from the headline---and among boundary subjects it is $55.8\%$. The boundary issue does not inflate the rejection rate. The EV and EU benchmarks, which are better-identified, also show majority rejection. Second, the benchmark assumes that choices are independent across menus conditional on the estimated parameters. In CPC18, subjects face repeated trials on the same lottery problems with feedback, which may induce serial dependence or learning. We address this directly in Section~\ref{subsec:robustness} below with a problem-level block bootstrap.

\subsection{Compression under library reduction}\label{subsec:frontier_empirical}

If the compression result depended on the specific composition of the eight-rule library, it would be much less compelling as evidence for the compressibility of behavior. We therefore ask how many of the eight rules are actually needed to preserve the concentration result, and which rules they are. Throughout this subsection, the attention rules \textsc{A1} and \textsc{A2} are always retained (they ensure feasibility), and a sublibrary is determined by which of the remaining six rules---\textsc{MMn}, \textsc{MMx}, \textsc{MAP}, \textsc{SAL}, \textsc{REG}, \textsc{DIS}---it includes. For each subject $i$ and each subset $S\subseteq\{\textsc{MMn},\textsc{MMx},\textsc{MAP},\textsc{SAL},\textsc{REG},\textsc{DIS}\}$, let $M_i(S)$ denote the MRCI computed on the library $S\cup\{\textsc{A1},\textsc{A2}\}$. Let $\bar M(S)$ denote the cross-subject average. The mechanical floor is the special case in which $S$ is empty, denoted $\bar M(\varnothing)$. The \emph{excess retention} of $S$ is
\[
R(S)\ =\ \frac{\bar M(S) - \bar M(\varnothing)}{\bar M(\mathcal F) - \bar M(\varnothing)},
\]
the fraction of the concentration above the mechanical floor that $S$ preserves. For each $K=0,1,\dots,6$, we select the $K$-rule subset $S^\star_K$ that maximizes $\bar M(\cdot)$ on a training sample (50\% of subjects) and evaluate $R(S^\star_K)$ out of sample on the held-out subjects, repeating over $100$ random splits to assess stability. Selection uses only internal concentration criteria---it does not optimize benchmark rejections---so the exercise is a genuine internal robustness check rather than a second specification search.

\paragraph{Set-based selection versus sequential addition.} For each $K$ the procedure selects the best $K$-rule subset by argmax over all $\binom{6}{K}$ subsets; order of inclusion plays no role in selection, and the $K$-rule subset is not constrained to extend the $(K-1)$-rule subset. Where the narrative below reads sequentially (``salience alone preserves about $40\%$; adding modal-payoff focusing reaches $71\%$''), it describes differences across set sizes, not a nested addition procedure. The distinction matters: in linear regression analysis, attribution of a coefficient change to individual covariates is sensitive to the order in which they are introduced when the procedure is sequential, even though the final fit is order-invariant, a point developed at length by \citet{gelbach2016}. Our set-based selection sidesteps this concern at the level of the headline retention numbers --- each row of Table~\ref{tab:rule_frontier} reports the best attainable $K$-rule retention, independent of which $(K-1)$-rule subset is nested inside it. Where the best-$K$ subsets do happen to nest --- salience at $K=1$, salience plus modal-payoff focusing at $K=2$, those two plus regret at $K=3$ --- the nesting is itself a property of the data rather than a procedural artifact: a small recurring core wins the cross-subject competition across $K$.

\medskip
Table~\ref{tab:rule_frontier} reports the results. Salience (\textsc{SAL}) is the single best rule at $K=1$ (selected in $100\%$ of splits), capturing about $40\%$ of the excess concentration. At $K=2$, the best pair is \textsc{MAP}~+~\textsc{SAL} ($95\%$ of splits), preserving about $71\%$ of the excess. At $K=3$, regret (\textsc{REG}) is added in $80\%$ of splits, raising retention to $82\%$. Reaching $90\%$ excess retention requires four rules: \textsc{MAP}, \textsc{SAL}, \textsc{REG}, and best-case extremum (\textsc{MMx}). Disappointment is never selected at $K\le 3$ and contributes the residual gap to full retention.

A small subset of rules wins the cross-subject MRCI competition. Salience and modal-payoff focusing rank as the top two rules in essentially every random split; regret comes third. We emphasize that this is a statement about which rules in our library best concentrate the cross-subject MRCI, not a claim that subjects in fact use those rules: the MRCI optimum is a rationalizability statistic, and the cross-subject competition selects rules that align well with observed per-problem choice frequencies. The qualifier matters in another sense too: the headline ``a small core captures most of the excess'' is $90\%$ of the average excess of about $0.028$---on the order of $0.025$ in absolute MRCI units. Large in proportional terms; small in absolute terms.

The reduced-library frontier above selects on a training sample of \emph{subjects} and evaluates retention on held-out subjects. As a stronger generalization check, we also hold out \emph{problems}: for each of $50$ random splits of the CPC18 problems into training and held-out halves, we select the $K$-rule subset that maximizes cross-subject MRCI on training-problem trials and evaluate its excess retention on held-out-problem trials. Generalization is somewhat noisier at low $K$, but the recurring core is consistent with the subject-split frontier. At $K=1$, \textsc{SAL} is the most consistently selected single rule ($60\%$ of splits, with \textsc{MMn} and \textsc{MAP} each appearing in a minority); from $K=2$ on, \textsc{MAP} and \textsc{SAL} are the recurring core (\textsc{MAP} in $94$--$96\%$ and \textsc{SAL} in $70$--$74\%$ of splits at $K=2$--$3$), with \textsc{REG} and \textsc{MMn} appearing as the next-most-selected rules. From $K=4$ the selected set reliably adds \textsc{MMn} and \textsc{MMx}; \textsc{DIS} is never selected at $K\le 5$. Four rules reach approximately $90\%$ retention on held-out problems. Detail is in Appendix~\ref{app:additional}, Table~\ref{tab:heldout_menu}.

\subsection{Robustness}\label{subsec:robustness}

The baseline analysis raises several natural concerns. We address them in two passes. First, two diagnostics whose results are central to interpreting the headline finding - heterogeneity across the A1/A2 dimension, and dependence-robustness given the repeated-trial structure with feedback - are reported here in full. Second, a battery of secondary diagnostics - bootstrap size, alternative activation thresholds, left-right relabeling, CPT boundary conditioning, the marginal information of parametric over permutation benchmarks - are summarized in one paragraph each, with detailed tables in Appendix~\ref{app:additional}. The headline findings survive all of these checks.

\paragraph{A1/A2 defaults, rule overlap, and the mechanical floor.}
A natural concern is that the raw MRCI is largely a mechanical artifact of the attention rules: every binary menu admits exactly one of A1, A2, so the always-feasible attention-only assignment achieves $\mathrm{HHI}=\alpha^2+(1-\alpha)^2$. The substantive object is therefore the \emph{excess} MRCI above this floor, which is solver-invariant (both MRCI and the floor are scalars defined without reference to a particular optimizer). The cross-subject mean excess is $0.028$ on a $[0,1]$ scale, with the 90th percentile at $0.085$; the broader concern that ``high raw MRCI is unsurprising given coarse-rule overlap'' is precisely why the benchmark test of Section~\ref{subsec:model_benchmark} is essential. The benchmark models also generate choices with nontrivial rule overlap, and the fact that observed MRCI still exceeds the benchmark for a majority of subjects means the data are more compressible than what utility models would mechanically produce. On average, about $4$ rules out of the eight in the library are admissible per menu, so coverage is broad enough that no single rule is forced on the optimum by mechanical inadmissibility.

The composition is heterogeneous, however, and the substantive few-rule signal is concentrated in the subjects whose data admit the most non-attention compression. The natural way to bin subjects on this dimension is by \emph{excess MRCI} (observed MRCI minus the subject-specific mechanical floor $\alpha^2+(1-\alpha)^2$): both quantities are scalars defined without reference to a specific MRCI-maximizing assignment, so the binning is solver-invariant. Partitioning the sample into excess-MRCI quartiles (Table~\ref{tab:a1a2_quartiles}) yields a sharp pattern: Q1 subjects (mean excess $0.090$) reject CPT at $89.5\%$ and BEAST at $79.7\%$ at the $5\%$ level, Q2 subjects (mean excess $0.028$) reject at $77.3\%$ and $63.4\%$, while Q3 subjects (mean excess $0.004$, essentially at the floor) reject at $28.7\%$ and $29.2\%$. Q4 contains subjects whose observed MRCI is at or, due to greedy-solver noise, marginally below the mechanical floor; their rejection rates ($50.3\%$ CPT, $39.2\%$ BEAST) reflect a residual left/right structure that the parametric benchmarks do not generate even when MRCI itself is at the floor. The interpretation is twofold. First, the substantive few-rule pattern documented in Section~\ref{subsec:frontier_empirical} is meaningful for the half of the sample whose data admit excess compression beyond one-sided attention; for the lower half, the MRCI is essentially the mechanical floor and the few-rule story does not apply. Second, even floor-level subjects reject CPT and BEAST at nontrivial rates, reflecting structure in the choice patterns themselves that the parametric benchmarks with i.i.d.\ noise cannot reproduce. 

\paragraph{Block-bootstrapped benchmark.}
Section~\ref{para:per_problem} settles the data side: MRCI is invariant to within-game choice dynamics. The benchmark side is a separate question. The standard test of Section~\ref{subsec:model_benchmark} fixes the menu sequence at the subject's actual menu sequence and randomizes only the per-trial choice indicator via Bernoulli$(q_t(\hat\theta))$; each of the $G\approx 25$ distinct games keeps its 25 trials in every synthetic dataset. To check whether this fixed-game-composition assumption inflates the CPT rejection rate, we run a problem-level block bootstrap: for each synthetic dataset, we resample game IDs with replacement (so a few games may appear twice or three times in the bootstrap menu sequence, others not at all) and draw fresh iid Bernoulli at each resampled trial index with the CPT-predicted $q$ for that game. Because $G$ is small, the bootstrap-induced variation in game multiplicities widens the simulated MRCI distribution substantially. Under this null, $22.7\%$ of subjects reject at the $5\%$ level and $10.9\%$ at the $1\%$ level (Table~\ref{tab:dependence_robust}), against $61.5\%$ and $52.9\%$ under iid CPT. The two rates bracket a methodological range: the standard rate is conditional on the realized sample of games, the bootstrap rate treats that sample as one draw from a stochastic process. We report both. For roughly one in five CPC18 subjects, observed compression exceeds what CPT generates under either simulator.

\paragraph{Secondary diagnostics.}
A battery of further checks confirms that the headline finding is not driven by auxiliary modeling choices and that rejection rates are not inflated by parameter-estimation noise. A Monte Carlo simulation under the CPT null with $M=200$ datasets per subject across 40 subjects shows the parametric bootstrap is conservative at all three nominal levels (empirical size $0.005$ at $\eta=0.01$, $0.029$ at $\eta=0.05$, $0.063$ at $\eta=0.10$; Table~\ref{tab:bootstrap_size}); rejection rates against CPT are therefore not inflated by parameter estimation. About $38\%$ of subjects have an at-the-boundary CPT estimate ($\hat\alpha\to 1$, the linear-utility limit); conditioning on subjects with both parameters interior, the CPT rejection rate at $\eta=0.01$ is $50.5\%$, essentially unchanged from the headline (Table~\ref{tab:cpt_boundary}). 

Random relabeling of left and right within each menu drops average MRCI from $0.545$ to $0.502$, close to the floor, with correlation $-0.16$ between observed and relabeled MRCI (Table~\ref{tab:rob_leftright}); the compression is tied to the actual choice patterns. Increasing the strict-comparison threshold $\varepsilon$ on the representative-payoff gap from $0$ to $0.10$, $0.20$, or $0.30$ leaves average MRCI at $0.546$ and the active-set size at about $4.0$ rules per menu (Table~\ref{tab:rob_activation}); the activation criterion is not on a knife-edge. Additionally, on 100 subjects, MRCI computed with $50$ vs.\ $200$ greedy restarts has correlation $0.9997$ and maximum absolute difference $0.012$, well below the empirical gaps that drive the main findings. Finally, of subjects who reject the permutation benchmark at $\eta=0.01$, $65\%$ also reject CPT and $53\%$ reject BEAST; of those who pass the permutation benchmark, $33\%$ still reject CPT and $31\%$ still reject BEAST (Table~\ref{tab:perm_vs_param}). The two benchmarks convey overlapping but distinct information: the permutation test catches a layer of structure that the parametric benchmarks can in part rationalize, and the parametric benchmarks catch a layer of structure that menu-independent randomness cannot.

\subsection{Replication and Cross-Environment Portability}\label{subsec:external}

To assess whether the compression logic survives across environments, we apply the same framework---same eight-rule library, same activation discipline, same benchmark specifications---to two additional datasets (Table~\ref{tab:cross_dataset}).

\paragraph{CPC15.} The CPC15 dataset \citep{erev_ert_plonsky_2017} is the predecessor to CPC18, sharing the same repeated-choice-with-feedback design but using a different set of 131 lottery problems. We pool three experiments (279 subjects, 440--1,660 observations per subject after filtering to risk problems with feedback). The mean MRCI is $0.525$ and the mechanical floor averages $0.519$, leaving substantive excess of only $0.006$---an order of magnitude smaller than in CPC18. At this magnitude there is essentially no compression to detect: the framework is portable in structure---the eight-rule library, the activation discipline, and the benchmark machinery all carry over without modification---but the substantive few-rule signal is at the threshold of detectability. The benchmarks still reject for a substantial share of subjects ($52\%$ reject CPT at $5\%$, $68\%$ reject EV).

\paragraph{BBS.} The BBS\ dataset \citep{baillon2020} is a sharply different environment: 139 subjects each make 70 one-shot incentivized binary lottery choices with no repetition or feedback, and lotteries have one to four outcomes. The mean MRCI is $0.597$ with mechanical floor $0.582$ and excess $0.015$. Benchmark rejection rates are lower than on CPC18 ($14\%$ reject CPT at $5\%$), reflecting both the small sample size per subject and the different design. 


\paragraph{More demanding Benchmarks.} On the BBS dataset we also implement a behavioral benchmark for reference-dependent choice: a multi-reference specification close to the model characterized in \citet{baillon2020}, where the authors allows for six alternative reference points (status quo, MaxMin, MinMax, X-at-Max-$P$, expected value, prospect itself) and identifies the one best supported by each subject's choices. The full specification and our per-subject estimation procedure are in Appendix~\ref{app:bbs_rd}. Table~\ref{tab:baillon_rd} reports the result. Only $1.4\%$ of subjects reject this benchmark at the $5\%$ level (and $3.6\%$ at the $10\%$ level)---an order of magnitude below the $14\%$ rejection rate against generic CPT. We read the very low rejection rate as evidence that, on the BBS dataset, a reference-dependent model with subject-specific reference points accounts for much of what our rule library captures.

\label{subsec:beast_cpc18}The analogous demanding benchmark for CPC18 is BEAST \citep{erev_ert_plonsky_2017}, the canonical predictive model of risky choice and the baseline of the CPC15/CPC18 prediction competitions. We use it verbatim from the official simulator distributed with \citet{plonsky_etal2025_nhb}, calibrated at the published population-level values; details of the model, the calibrated parameters, and how we generate per-menu choice probabilities for the MRCI Monte Carlo are in Appendix~\ref{app:beast}.

Table~\ref{tab:beast_cpc18} reports the result. $52.9\%$ of subjects reject BEAST at the $5\%$ level and $44.8\%$ at the $1\%$ level, lower than the $61.5\%$ ($52.9\%$) rates against generic CPT but still close to a majority of subjects. The reduction relative to CPT is informative: BEAST absorbs part of the few-rule signal that generic CPT cannot. 

\paragraph{Cross-dataset summary.}
Across three datasets spanning different problem sets, designs, and subject populations, the broad qualitative pattern is preserved: raw MRCI is high, the attention-rule mechanical floor accounts for most of it, and the excess above the floor---though modest---is large enough to be visible against stochastic-utility benchmarks for a meaningful share of subjects. The strength of this signal is environment-dependent: it is strongest in CPC18 (large repeated dataset with feedback), substantially weaker in CPC15 (excess of $0.006$), and in between in BBS (richer one-shot lotteries but only $70$ choices per subject). Two design features plausibly drive the magnitude gap. CPC18's problem set was deliberately curated to elicit a broad span of risky-choice anomalies---Allais paradoxes, certainty-like effects, break-even effects, splitting effects, and rare-event weighting---which raises the alignment between subjects' per-problem choice frequencies and the comparison primitives in our rule library. CPC15 uses a less anomaly-focused problem set, narrowing the band of behavioral phenomena the framework can detect, while BBS contributes only $70$ one-shot decisions per subject, limiting per-subject statistical power even though the lotteries are richer in outcome structure. The framework is portable in structure, but its substantive output is bounded by what each design probes. 


\section{Conclusion}\label{sec:conclusion}

Behavioral theories rests on a working premise: a small number of named mechanisms can organize a wide range of risky choices. The empirical force of the enterprise depends on whether behavior is in fact parsimonious in this sense. This paper asks the question quantitatively. We define a non-parametric index, the Maximum Rule Concentration Index, that measures how tightly a dataset of risky choices can be organized by a fixed library of simple decision rules drawn from canonical theories. The index is a property of the realized data, prior to any model fit; the question is then whether the data exhibit such parsimony, and whether conventional utility models reproduce it.

The answer in three lottery-choice datasets is: yes, the data are parsimonious, and yes, that parsimony is informative. The substantive concentration above what one-sided attention rules generate mechanically is small in absolute terms - on the order of three percent of the index's range - but it is statistically robust. For a majority of subjects, the observed concentration exceeds what expected utility, cumulative prospect theory, and the BEAST model with their standard noise structures generate on the same menus.

The compression concentrates on a small recurring subset of the library: salience-based comparisons, modal-payoff focusing, and regret-type comparisons are the three rules that win the cross-subject MRCI competition across random training/test splits, with about ninety percent of the excess captured by four rules. We do not read this as identification of the mechanisms subjects in fact use - the index is a rationalizability statistic, not a structural estimate - but as a statement about which rules in our maintained library most concentrate the data. The signal is heterogeneous across subjects in a structured way: it is essentially absent in subjects whose choices are dominated by one-sided attention to a single option, and concentrated in the others. 

Across the three datasets the framework is portable in structure but not in magnitude: the few-rule signal is strongest where the experimental design probes a richer space of behavioral phenomena and weakest where the design admits fewer such phenomena to begin with. To the extent that the explanatory tradition of behavioral economics presupposes parsimony of behavior as a phenomenon, we read these findings as a vindication of that presupposition as an empirical matter, rather than as an endorsement of any one theory: the index is mechanism-coarse and does not identify the model people actually follow.

A first limitation is that all objects in the paper are conditional on a maintained library of decision rules; the index measures revealed compressibility, not latent rule usage. The library is deliberately coarse and parameter-free, which makes its findings transparent but does not exhaust the space of plausible mechanisms. A second limitation is that the magnitude of observed parsimony is environment-dependent: experimental designs in behavioral economics often co-evolve with the theories they test, and the parsimony we measure on a given dataset reflects both the realized choices and the menu space those choices were elicited over.
\bibliographystyle{plainnat}
\bibliography{bibliography}


\clearpage

\begin{center}
{\Large\bfseries Tables and Figures}
\end{center}

\begin{figure}[htbp]
    \centering
    \includegraphics[width=1.0\linewidth]{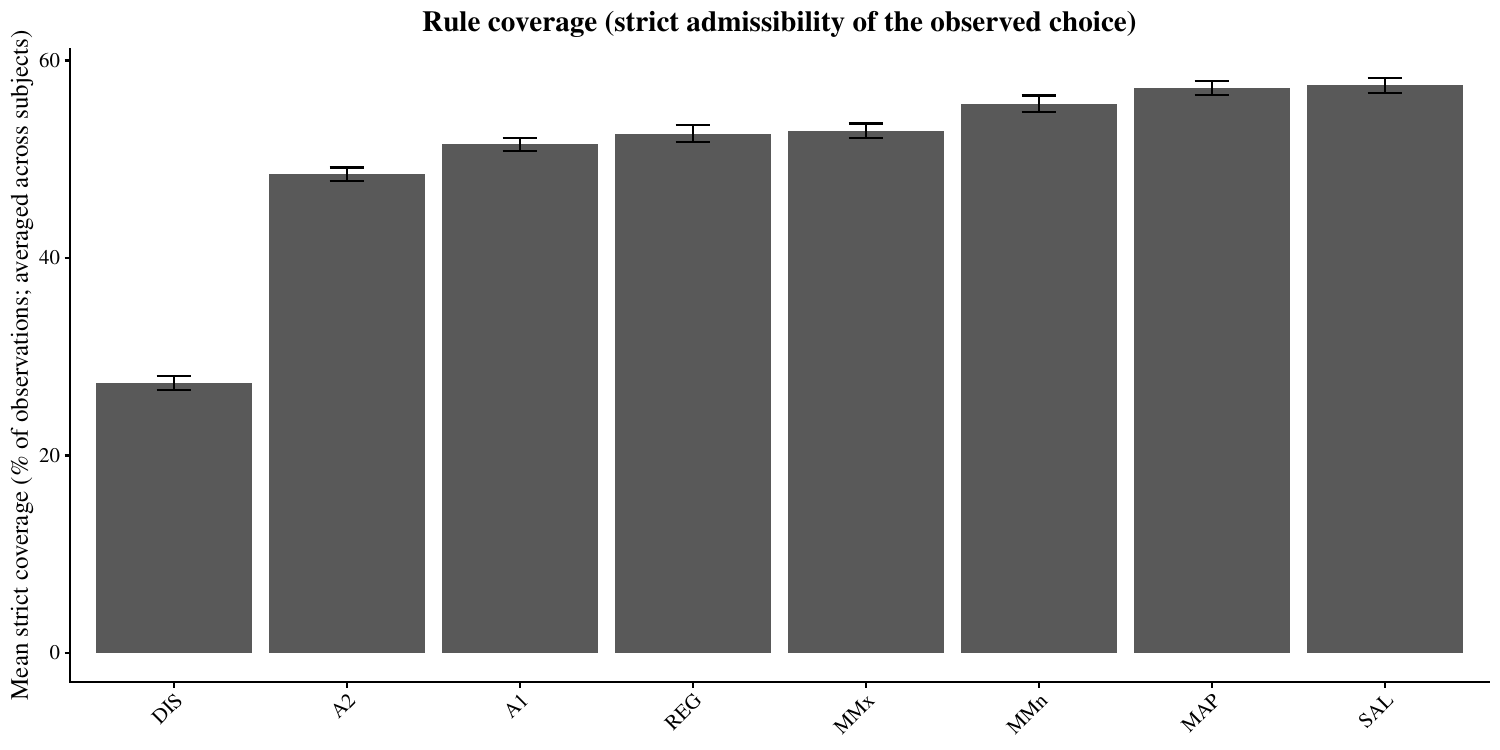}
    \caption{Mean activity coverage by rule. Coverage is the share of a subject's menus at which rule $f$ delivers a strict recommendation that coincides with the realized choice. Cross-subject means with standard error bars.}
    \label{fig:rule_coverage}
\end{figure}

\begin{figure}[htbp]
    \centering
    \begin{subfigure}[b]{0.48\textwidth}
    \includegraphics[width=\linewidth]{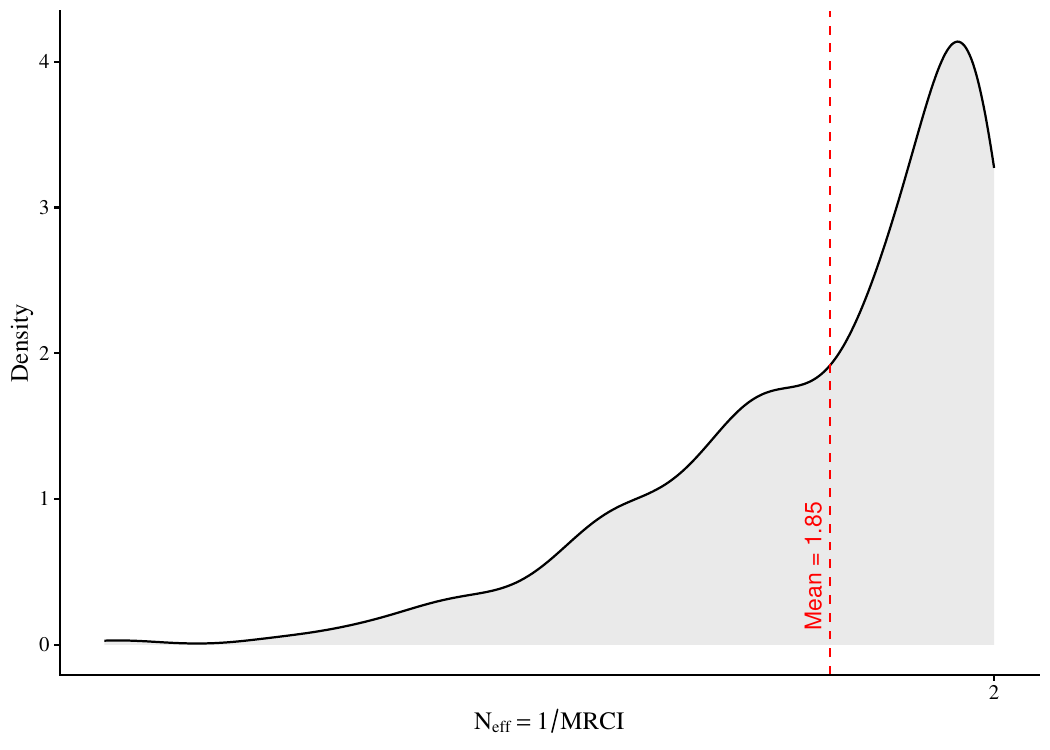}
    \caption{Distribution of $N_{\mathrm{eff}} = 1/\text{MRCI}$}
    \label{fig:neff}
    \end{subfigure}
    \hfill
    \begin{subfigure}[b]{0.48\textwidth}
    \includegraphics[width=\linewidth]{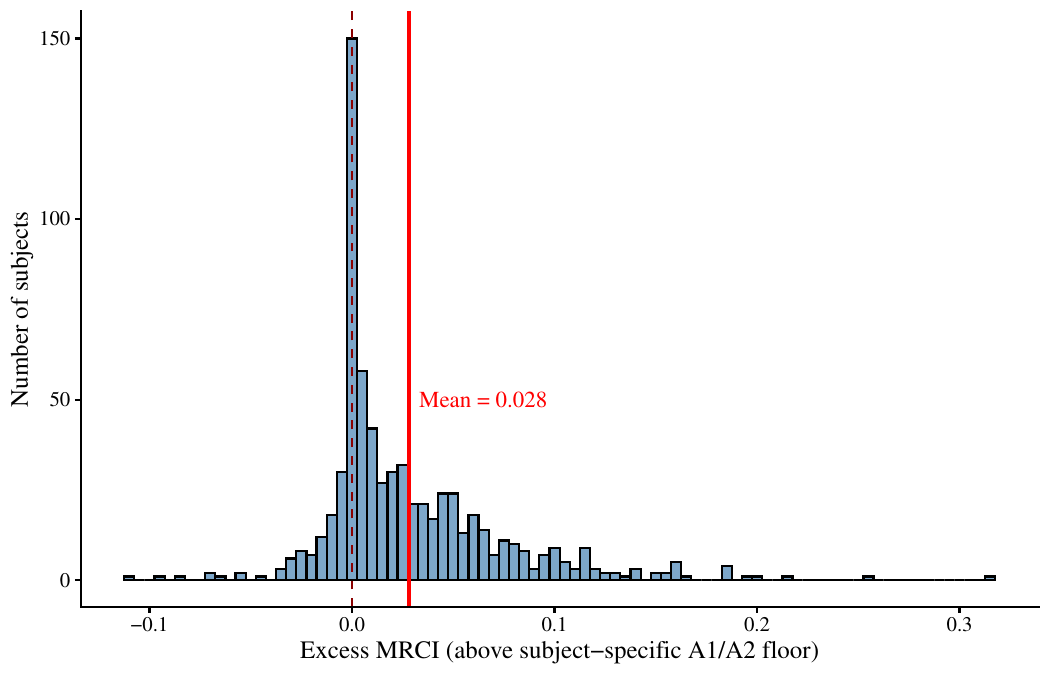}
    \caption{Excess MRCI above attention-rule floor}
    \label{fig:excess}
    \end{subfigure}
    \caption{Distribution of behavioral compression across subjects. Panel (a): effective number of rules. Panel (b): excess MRCI, defined as MRCI minus the subject-specific mechanical floor $\alpha^2 + (1-\alpha)^2$.}
\end{figure}

\begin{figure}[htbp]
    \centering
    \includegraphics[width=1.0\linewidth]{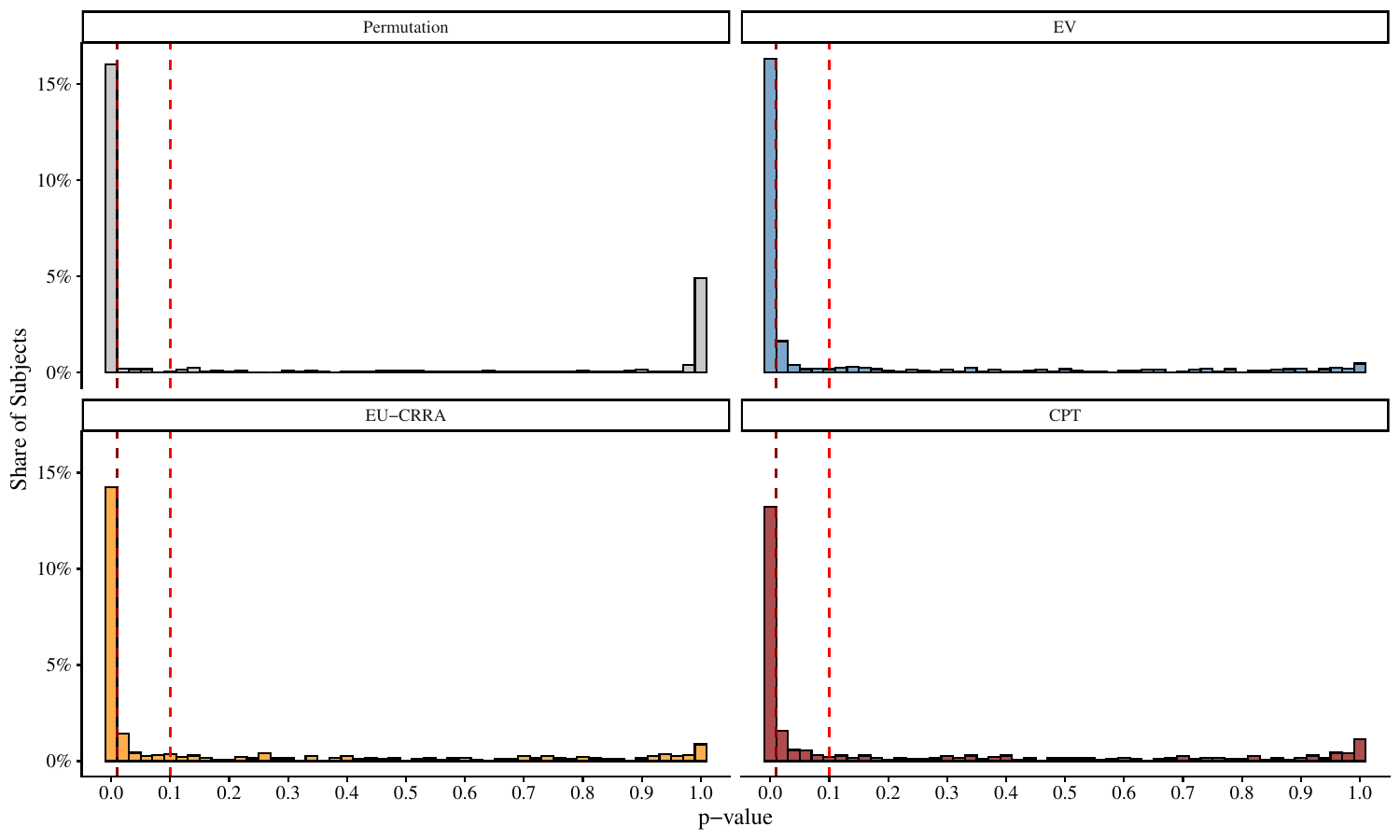}
    \caption{Distribution of $p$-values from the benchmark tests under the full baseline library. Each panel shows the histogram of subject-level $p$-values for one benchmark. Dashed lines indicate the $1\%$ and $10\%$ significance thresholds.}
    \label{fig:test}
\end{figure}

\clearpage

\input{Table_1_Summary_Stats}

\input{Table_BootstrapCI}

\input{Table_2_Benchmark_Results}

\clearpage

\input{Table_Rule_Frontier}

\input{Table_A1A2_Quartiles}

\input{Table_Dependence_Robust}

\clearpage

\input{Table_CrossDataset}

\input{Table_BaillonRD}

\input{Table_BEAST_CPC18}

\clearpage

\appendix
\setcounter{page}{1}
\renewcommand{\thepage}{\arabic{page}}
\numberwithin{table}{section}
\numberwithin{figure}{section}
\numberwithin{equation}{section}
\input{appendix}

\end{document}

%% file: Table_1_Summary_Stats.tex
\begin{table}[htbp]
\centering
\caption{Summary Statistics}
\label{tab:summary_stats}
\begin{threeparttable}
\begin{tabular}{lc}
\toprule
Statistic & Value \\
\midrule
\multicolumn{2}{l}{\textit{Demographics}} \\
Number of Subjects & 686 \\
Female Share & 0.54 \\
Age (Mean) & 24.9 \\
Age (Std. Dev.) & 2.9 \\
\addlinespace
\multicolumn{2}{l}{\textit{Rule Concentration}} \\
MRCI (Mean) & 0.545 \\
MRCI (Std. Dev.) & 0.050 \\
$N_{\mathrm{eff}}$ (Mean) & 1.85 \\
Share Single Rule ($N_{\mathrm{eff}} \approx 1$) & 0.000 \\
\bottomrule
\end{tabular}
\begin{tablenotes}
\small
\item \textit{Notes:} $N=686$ subjects from CPC18 (risk problems, feedback condition).
MRCI computed with greedy heuristic ($100$ restarts); $N_{\mathrm{eff}}=1/\mathrm{MRCI}$.
\end{tablenotes}
\end{threeparttable}
\end{table}

%% file: Table_BootstrapCI.tex
\begin{table}[htbp]
\centering
\caption{Bootstrap Confidence Intervals for Mean Excess MRCI}
\label{tab:bootstrap_excess}
\begin{threeparttable}
\begin{tabular}{lcccccc}
\toprule
Dataset & $N$ & Mean excess & 95\% CI & SE & Share $>0$ & 90th pct. \\
\midrule
CPC18 & 686 & 0.0281 & [0.0248, 0.0315] & 0.0017 & 0.72 & 0.0850 \\
CPC15 & 279 & 0.0058 & [0.0028, 0.0091] & 0.0016 & 0.53 & 0.0306 \\
BBS & 139 & 0.0153 & [0.0107, 0.0203] & 0.0025 & 0.42 & 0.0499 \\
\bottomrule
\end{tabular}
\begin{tablenotes}
\small
\item \textit{Notes:} Nonparametric subject-level bootstrap with $B=2000$ resamples. Excess MRCI is MRCI minus subject-specific mechanical floor $\alpha^2 + (1-\alpha)^2$.
\end{tablenotes}
\end{threeparttable}
\end{table}

%% file: Table_2_Benchmark_Results.tex
\begin{table}[htbp]
\centering
\caption{Benchmark Test Results: Share of Subjects Rejecting Each Benchmark}
\label{tab:benchmark_results}
\begin{threeparttable}
\begin{tabular}{lccc}
\toprule
Benchmark & $p \le 0.10$ & $p \le 0.05$ & $p \le 0.01$ \\
\midrule
Permutation & 0.665 & 0.656 & 0.633 \\
EV (risk-neutral) & 0.754 & 0.735 & 0.653 \\
EU-CRRA & 0.673 & 0.644 & 0.570 \\
CPT & 0.655 & 0.615 & 0.529 \\
\bottomrule
\end{tabular}
\begin{tablenotes}
\small
\item \textit{Notes:} Each cell reports the fraction of subjects for whom the observed MRCI
exceeds the benchmark distribution at the indicated significance level.
Permutation: nonparametric test against menu-independent randomness.
EV, EU-CRRA, CPT: model-based benchmarks estimated subject-by-subject by maximum likelihood.
Monte Carlo with $B=500$ synthetic draws per benchmark.
\end{tablenotes}
\end{threeparttable}
\end{table}

%% file: Table_Rule_Frontier.tex
\begin{table}[htbp]
\centering
\caption{Rule-Level Frontier: Excess Retention and Selected Rules}
\label{tab:rule_frontier}
\begin{threeparttable}
\begin{tabular}{cccc}
\toprule
$K$ (rules) & MRCI & $R_K$ & Top selections (frequency) \\
\midrule
0 & 0.517 & 0.000 & --- \\
1 & 0.528 & 0.398 & SAL (100\%) \\
2 & 0.537 & 0.707 & MAP (100\%), SAL (95\%), REG (5\%) \\
3 & 0.540 & 0.822 & MAP (100\%), SAL (100\%), REG (80\%) \\
4 & 0.542 & 0.895 & MAP (100\%), SAL (100\%), REG (85\%) \\
5 & 0.545 & 0.999 & MAP (100\%), MMn (100\%), MMx (100\%) \\
6 & 0.545 & 1.000 & DIS (100\%), MAP (100\%), MMn (100\%) \\
\bottomrule
\end{tabular}
\begin{tablenotes}
\small
\item \textit{Notes:} $K$ is the number of rules from $\{\textsc{MMn},\textsc{MMx},\textsc{MAP},\textsc{SAL},\textsc{REG},\textsc{DIS}\}$ included alongside the limited-attention rules \textsc{A1} and \textsc{A2}. Selection performed on a 50\% training sample over $N_{\text{splits}}=100$ random splits; excess retention $R_K$ evaluated out of sample. ``Top selections'' lists the three most frequently chosen rules at each $K$.
\end{tablenotes}
\end{threeparttable}
\end{table}

%% file: Table_A1A2_Quartiles.tex
\begin{table}[ht!]\centering
\renewcommand{\arraystretch}{1.10}\small
\caption{Excess-MRCI quartile analysis on CPC18.}
\label{tab:a1a2_quartiles}
\begin{threeparttable}
\begin{tabular}{lcccc}
\toprule
& Q1 (high excess) & Q2 & Q3 & Q4 (low excess) \\
\midrule
Excess MRCI range & [0.047,\,0.314] & [0.013,\,0.047] & [0.000,\,0.012] & [-0.111,\,0.000] \\
Mean excess MRCI & 0.090 & 0.028 & 0.004 & -0.010 \\
Mean raw MRCI & 0.606 & 0.541 & 0.515 & 0.519 \\
$N$ subjects & 172 & 172 & 171 & 171 \\
\midrule
Reject EV ($\eta=0.05$) & 97.1\% & 95.9\% & 46.8\% & 53.8\% \\
Reject EU ($\eta=0.05$) & 91.9\% & 88.4\% & 29.8\% & 47.4\% \\
Reject CPT ($\eta=0.05$) & 89.5\% & 77.3\% & 28.7\% & 50.3\% \\
Reject CPT ($\eta=0.01$) & 82.6\% & 67.4\% & 19.9\% & 41.5\% \\
Reject BEAST ($\eta=0.05$) & 79.7\% & 63.4\% & 29.2\% & 39.2\% \\
Reject BEAST ($\eta=0.01$) & 76.2\% & 50.6\% & 21.6\% & 30.4\% \\
\bottomrule
\end{tabular}
\begin{tablenotes}\footnotesize
\item Quartiles are formed on subject-level excess MRCI, where excess MRCI = observed MRCI minus the subject-specific mechanical floor $\alpha^2+(1-\alpha)^2$. Both quantities are scalars defined without reference to a specific MRCI-maximizing assignment, so the binning is solver-invariant. Q1 is the highest-excess quartile (subjects whose data admit the strongest non-attention compression); Q4 is the lowest-excess quartile (subjects whose MRCI is essentially at the mechanical floor). The BEAST benchmark is the canonical model of \citet{erev_ert_plonsky_2017}, implemented from the official simulator.
\end{tablenotes}
\end{threeparttable}
\end{table}

%% file: Table_Dependence_Robust.tex
\begin{table}[htbp]
\centering
\caption{Dependence-Robust Benchmark Rejection Rates (CPC18)}
\label{tab:dependence_robust}
\begin{threeparttable}
\begin{tabular}{lccc}
\toprule
Benchmark & $N$ & $p\le 0.05$ & $p\le 0.01$ \\
\midrule
CPT block bootstrap & 686 & 0.227 & 0.109 \\
\bottomrule
\end{tabular}
\begin{tablenotes}
\small
\item \textit{Notes:} ``CPT block bootstrap'' simulates from the previously fit CPT model but resamples 25-trial game-level blocks with replacement to break i.i.d.-across-trials assumptions. $B = 100$ Monte Carlo draws.
\end{tablenotes}
\end{threeparttable}
\end{table}

%% file: Table_CrossDataset.tex
\begin{table}[htbp]
\centering
\caption{Cross-Dataset Comparison}
\label{tab:cross_dataset}
\begin{threeparttable}
\begin{tabular}{lccc}
\toprule
 & CPC18 & CPC15 & BBS \\
\midrule
Subjects & 686 & 279 & 139 \\
Trials/subject & 400--560 & 440--1{,}660 & 70 \\
Design & Repeated, feedback & Repeated, feedback & One-shot \\
\addlinespace
MRCI & 0.545 & 0.525 & 0.597 \\
$N_{\mathrm{eff}}$ & 1.85 & 1.91 & 1.71 \\
Mechanical floor & 0.517 & 0.519 & 0.582 \\
Excess MRCI & 0.0281 & 0.0058 & 0.0153 \\
Active rules/menu & 4.0 & 3.9 & 4.1 \\
\addlinespace
Perm.\ reject (5\%) & 0.656 & 0.523 & 0.367 \\
CPT reject (5\%) & 0.615 & 0.520 & 0.144 \\
\bottomrule
\end{tabular}
\begin{tablenotes}
\small
\item \textit{Notes:} Same eight-rule baseline library applied to all datasets. CPC15 pools three experiments (risk problems, feedback condition). CPC18 uses $B=500$; CPC15 and BBS use $B=50$.
\end{tablenotes}
\end{threeparttable}
\end{table}

%% file: Table_BaillonRD.tex
\begin{table}[htbp]
\centering
\caption{BBS Reference-Dependent Benchmark on the BBS Dataset}
\label{tab:baillon_rd}
\begin{threeparttable}
\begin{tabular}{lcc}
\toprule
Statistic & Value & \\
\midrule
Subjects & 139 &  \\
Mean MRCI (observed) & 0.597 & \\
Reject RD benchmark at $5\%$ & 0.014 & \\
Reject RD benchmark at $10\%$ & 0.036 & \\
\addlinespace
\multicolumn{3}{l}{\textit{Best-fitting reference rule per subject}} \\
Status Quo & 0.17 & \\
MaxMin & 0.22 & \\
MinMax & 0.12 & \\
X at Max $P$ & 0.11 & \\
Expected Value & 0.14 & \\
Prospect Itself & 0.26 & \\
\bottomrule
\end{tabular}
\begin{tablenotes}
\small
\item \textit{Notes:} Per-subject MLE of the \citet{baillon2020} reference-dependent model (their Eq.~6), with linear consumption utility, Tversky--Kahneman gain--loss utility (same curvature $\alpha$ for gains and losses, loss aversion $\lambda$), and Prelec probability weighting $w(p)=\exp(-(-\ln p)^\gamma)$ (same $\gamma$ for gains and losses); cumulative-from-above decision weights are applied to gain outcomes and cumulative-from-below to loss outcomes, per their Eq.~(1). Synthetic datasets simulated from the fitted choice probabilities ($B=100$). $p$-value is the right-tail Monte Carlo $p$-value as in \eqref{eq:pu_mc}. BBS estimate the same model via Bayesian hierarchical modeling; the present per-subject MLE differs methodologically but uses the same parametric form.
\end{tablenotes}
\end{threeparttable}
\end{table}

%% file: Table_BEAST_CPC18.tex
\begin{table}[ht!]\centering
\renewcommand{\arraystretch}{1.10}\small
\caption{BEAST benchmark on CPC18.}
\label{tab:beast_cpc18}
\begin{threeparttable}
\begin{tabular}{lcc}
\toprule
& BEAST & Generic CPT \\
\midrule
$N$ subjects & \multicolumn{2}{c}{686} \\
Reject at $\eta=0.05$ & 52.9\% & 61.5\% \\
Reject at $\eta=0.01$ & 44.8\% & 52.9\% \\
\midrule
Mean BEAST $\Pr(\text{left})$ & \multicolumn{2}{c}{0.505} \\
\bottomrule
\end{tabular}
\begin{tablenotes}\footnotesize
\item BEAST is the canonical CPC15/CPC18 baseline of \citet{erev_ert_plonsky_2017}, implemented using the official R simulator from the \citet{plonsky_etal2025_nhb} replication package. The model evaluates each prospect as the sum of (i) the best estimate of expected value, (ii) the mean of $\kappa$ samples drawn by one of four sampling tools (unbiased, equal-weighting, sign, contingent pessimism), and (iii) Gaussian noise with scale $\sigma$. Per-agent parameters are drawn uniformly from $(0,\sigma)$, $\{1,\dots,\kappa\}$, $(0,\beta)$, $(0,\theta)$, $(0,\gamma)$; we use the published calibrated values $\sigma=7$, $\kappa=3$, $\beta=2.6$, $\gamma=0.5$, $\theta=1$ \citep[p.~387]{erev_ert_plonsky_2017}. For each unique CPC18 menu we simulate 1,500 BEAST agents through the full 25-trial sequence and average across trials to obtain a menu-level choice probability $q_t$; the MRCI Monte Carlo benchmark of Section~\ref{subsec:model_benchmark} is then run with these $q_t$ values, drawing $B=500$ synthetic datasets per subject. BEAST is calibrated at the population level rather than fit per subject; this is faithful to the published model and provides an asymmetrically demanding benchmark---more demanding than per-subject CPT in the sense that it was designed and calibrated specifically for the CPC15/CPC18 problem set.
\end{tablenotes}
\end{threeparttable}
\end{table}

%% file: appendix.tex
\begin{LARGE}
    \noindent\textbf{Appendix}
\end{LARGE}

\section{Additional Tables and Figures}\label{app:additional}

\input{Table_CPT_Boundary}

\input{Table_HeldOut_Menu}

\input{Table_BootstrapSize}

\input{Table_Rob_LeftRight}

\input{Table_Rob_Activation}

\input{Table_Perm_vs_CPT}

\clearpage

\section{Rule Library: Formal Definitions}\label{ec:rules}

This section provides the formal definitions of the decision rules used in the analysis. The maintained baseline library contains eight rules: 
\[\mathcal F=\{\textsc{MMn}, \textsc{MMx}, \textsc{MAP}, \textsc{SAL}, \textsc{REG}, \textsc{DIS}, \textsc{A1}, \textsc{A2}\}.
\] The rules below are deliberately coarse, parameter-free implementations of mechanisms drawn from canonical behavioral theories of risky choice; they are not the mechanisms themselves. A given mechanism (e.g., salience theory of \citet{bordalo2012_risk}) may be parameterized in many ways, and our \textsc{SAL} rule corresponds to a particular parameter-free implementation of the salience comparison. The MRCI is therefore conditional on this specific library, and all empirical claims in the paper should be read as statements about how parsimoniously the data align with these rule prescriptions, not as identification of the underlying psychological mechanisms. Each rule $f$ maps an objective binary menu $(L^1,L^2)$ to a perceived menu $(\pi_f(L^1;L^2), \pi_f(L^2;L^1))$.

\paragraph{Common state-space construction.}
For \textsc{SAL} and \textsc{REG}, the rule prescription depends on a state-by-state comparison of the two lotteries and therefore requires a common state space on which $L^1$ and $L^2$ are represented as state-contingent acts. We use the product space induced by the independence of $L^1$ and $L^2$. If $L^i$ has support $\{z^i_1,\dots,z^i_{m_i}\}$ with probabilities $(\pi^i_1,\dots,\pi^i_{m_i})$, the common state space is $S=\{1,\dots,m_1\}\times\{1,\dots,m_2\}$ with state probabilities $p_{(k,\ell)}=\pi^1_k\,\pi^2_\ell$. Each lottery $L^i$ is the act $u^i:S\to\mathbb R$ with $u^1(k,\ell)=z^1_k$ and $u^2(k,\ell)=z^2_\ell$. The state-by-state comparisons used by \textsc{SAL} and \textsc{REG} below are with respect to this construction. For binary lotteries, $|S|=4$. The independence assumption matters: under a comonotonic pairing (e.g., the $k$-th worst outcome of $L^1$ paired with the $k$-th worst of $L^2$), \textsc{SAL} and \textsc{REG} would prescribe different recommendations on some menus. We adopt the independent pairing throughout, which is the standard convention in the salience and regret theory literatures and corresponds to evaluating the two lotteries as if they were resolved by independent random draws.

\paragraph{MaxMin (\textsc{MMn}).}
The decision maker replaces $L^i$ by its sure minimum payoff: $\pi_{\textsc{MMn}}(L^i;L^j)=(z^i_1 \text{ w.p.\ }1)$. This ``safety-first'' rule focuses exclusively on worst-case outcomes.

\paragraph{MinMax (\textsc{MMx}).}
The decision maker replaces $L^i$ by its sure maximum payoff: $\pi_{\textsc{MMx}}(L^i;L^j)=(z^i_{m_i} \text{ w.p.\ }1)$.

\paragraph{Max at $p$ (\textsc{MAP}).}
The decision maker replaces $L^i$ by the sure modal payoff of $L^i$ (the payoff with the highest probability; ties broken by selecting the highest payoff):
\[
\pi_{\textsc{MAP}}(L^i;L^j)=\big(z^i_{k^\star(L^i)} \text{ w.p.\ }1\big), \qquad k^\star(L^i)\in\arg\max\{z^i_k: \pi^i_k=\max_q \pi^i_q\}.
\]

\paragraph{Salience (\textsc{SAL}).}
On the common state space $S$ defined above, the salience of state $s$ is $\sigma_s(L^1,L^2):=|u^1(s)-u^2(s)|/(|u^1(s)|+|u^2(s)|+1)$. The rule collapses each lottery to its payoff in the most salient state:
\[
\pi_{\textsc{SAL}}(L^i;L^j) := \big(u^i(s^{\mathrm{sal}}) \text{ w.p.\ }1\big), \qquad s^{\mathrm{sal}} \in \arg\max_{s\in S} \sigma_s(L^1,L^2).
\]
Ties in $\sigma_s$ are broken by the largest absolute payoff difference $|u^1(s)-u^2(s)|$.

\paragraph{Regret (\textsc{REG}).}
On the same state space $S$, the worst foregone payoff for option $L^i$ is $\rho^i(L^1,L^2):=\max_{s\in S}\max\{u^j(s)-u^i(s),0\}$, the largest amount by which the alternative could exceed $L^i$ in any state. The rule maps $L^i$ to $\pi_{\textsc{REG}}(L^i;L^j):=(-\rho^i \text{ w.p.\ }1)$, so the decision maker minimizes the maximum foregone gain.

\paragraph{Disappointment (\textsc{DIS}).}
The reference payoff is the modal payoff of $L^i$: $\mathrm{ref}(L^i):=\max\{z:\Pr(u^i_t(s)=z)=\max_y\Pr(u^i_t(s)=y)\}$. The disappointment index is $D(L^i):=\max\{|(\mathrm{ref}(L^i)-z)/(|\mathrm{ref}(L^i)|+|z|+1)|: z<\mathrm{ref}(L^i)\}$, with $D(L^i)=0$ if no payoff is below the reference. The rule maps $\pi_{\textsc{DIS}}(L^i;L^j):=(-D(L^i) \text{ w.p.\ }1)$.

\paragraph{Attention to $L^1$ (\textsc{A1}) and Attention to $L^2$ (\textsc{A2}).}
Under \textsc{A1}, the second lottery is replaced by a dominated baseline: $\pi_{\textsc{A1}}(L^1;L^2)=L^1$ and $\pi_{\textsc{A1}}(L^2;L^1)=(-M \text{ w.p.\ }1)$, where $M$ is a large positive constant set below all payoffs in the data. Symmetrically, \textsc{A2} replaces $L^1$ by $-M$. In the empirical implementation we set $M = 10^6$, which is well below the minimum payoff in any of the datasets we use; results are insensitive to any choice of $M$ that lies below all observed payoffs. \textsc{A1} and \textsc{A2} encode one-sided limited consideration: the decision maker attends only to the left option (\textsc{A1}) or only to the right option (\textsc{A2}). They are substantive behavioral rules in their own right, capturing the limited-consideration channel of \citet{manzini2007,masatlioglu2012,manzini_mariotti2014}, while also ensuring that the activity set $\mathcal F_t^{\mathrm{act}}$ is non-empty at every binary menu. Any subject's MRCI is bounded below by the corresponding mechanical floor $\alpha^2 + (1-\alpha)^2$, where $\alpha$ is the subject's left-choice share.


\section{A Multi-Reference Reference-Dependent Benchmark (BBS)}\label{app:bbs_rd}

This appendix gives the  reference-dependent (RD) benchmark used in Section~\ref{subsec:external}.

\paragraph{Model.}
The model is inspired from the model characterized in \citet{baillon2020},  Equations~(1) and~(6). For a lottery $F$ with ordered outcomes $z_1<\cdots<z_m$ and probabilities $\pi_1,\dots,\pi_m$, a reference point $r$, and parameters $(\alpha,\lambda,\gamma)$,
\begin{equation}\label{eq:rd_baillon_app}
\mathrm{RD}(F;\,r,\alpha,\lambda,\gamma)\ =\ \sum_k z_k\,\pi_k\ +\ \mathrm{PT}_r(F;\alpha,\lambda,\gamma),
\end{equation}
where the first sum is linear consumption utility and $\mathrm{PT}_r$ is a Tversky--Kahneman gain--loss component that splits at $r$:
\begin{align*}
\mathrm{PT}_r(F;\alpha,\lambda,\gamma)\ =\ &\sum_{k:\,z_k\geq r}(z_k-r)^{\alpha}\,\big[w(P_{\geq k};\gamma)-w(P_{>k};\gamma)\big] \\
&-\lambda\sum_{k:\,z_k<r}(r-z_k)^{\alpha}\,\big[w(P_{\leq k};\gamma)-w(P_{<k};\gamma)\big].
\end{align*}
We follow \citet{baillon2020} in using the same curvature $\alpha\in(0,1]$ for gains and losses, the same Prelec probability-weighting function $w(p;\gamma)=\exp\!\big(-(-\ln p)^{\gamma}\big)$ for both sides, and a single loss-aversion parameter $\lambda>0$ on the loss side. The cumulative probabilities are $P_{\geq k}=\sum_{j\geq k}\pi_j$, $P_{>k}=\sum_{j>k}\pi_j$, $P_{\leq k}=\sum_{j\leq k}\pi_j$, $P_{<k}=\sum_{j<k}\pi_j$. Gain outcomes use cumulative-from-above decision weights, loss outcomes use cumulative-from-below decision weights, as in their Equation~(1). Choice probabilities follow a Luce rule with precision $\xi>0$:
\[
\Pr(F\succ G)\ =\ \frac{1}{1+\exp\!\big(\xi[\mathrm{RD}(G)-\mathrm{RD}(F)]\big)}.
\]

\paragraph{Reference-point rules.}
The reference point is selected per subject from \citet[Table~1]{baillon2020}'s six-rule taxonomy. Five rules pin $r$ deterministically:
\begin{itemize}\itemsep0.1em
\item \emph{Status Quo}: $r=0$ (the participation fee baseline; all outcomes in the BBS experiment are gains under this rule, see \citealp[p.~96]{baillon2020}).
\item \emph{MaxMin}: $r=\max(\min L^1,\min L^2)$, the maximum of the minimum outcomes of the two prospects.
\item \emph{MinMax}: $r=\min(\max L^1,\max L^2)$, the minimum of the maximum outcomes.
\item \emph{X at Max~$P$}: $r$ equals the highest-probability outcome across the menu.
\item \emph{Expected Value} (prospect-specific): $r=E[F]$ for each prospect $F$.
\end{itemize}
The sixth rule, \emph{Prospect Itself}, is the only stochastic reference: the reference distribution is the prospect itself, $R=F$, so the gain--loss component in~\eqref{eq:rd_baillon_app} integrates over $F$ rather than being evaluated at a single point:
\[
\mathrm{RD}_{\mathrm{PI}}(F)\ =\ \sum_k z_k\,\pi_k\ +\ \sum_k \pi_k\,\mathrm{PT}_{z_k}(F;\alpha,\lambda,\gamma).
\]

\paragraph{Estimation and selection.}
We estimate $(\alpha,\lambda,\gamma,\xi)$ by per-subject maximum likelihood, in contrast to the Bayesian hierarchical estimation of \citet{baillon2020}.  Numerically, $\alpha\in(0.1,1.5)$ and $\gamma\in(0.05,2.0)$ are reparameterized through affine logistic transforms, $\lambda,\xi>0$ via exponentials. Six starting points spanning the parameter space are used, and the Nelder--Mead optimum is taken across all of them.

Selection of the per-subject reference-point rule is by AIC, $2\,\mathrm{NLL}+2k_{\mathrm{eff}}$, rather than by raw NLL. Under Status Quo with $r=0$ and the BBS gain-only outcomes, no outcomes are losses and the loss-aversion parameter $\lambda$ does not enter the likelihood; the effective parameter count for that rule is therefore $k_{\mathrm{eff}}=3$, against $k_{\mathrm{eff}}=4$ for the other five rules. Picking by raw NLL would systematically reward the over-parameterized rules. The Bayesian hierarchical approach achieves an analogous complexity penalty implicitly through the marginal likelihood.

The resulting per-subject classification ranks Prospect Itself ($26\%$), MaxMin ($22\%$), and Status Quo ($17\%$) as the three most frequent rules in our sample of 139 subjects, with the remaining mass spread across Expected Value ($14\%$), MinMax ($12\%$), and X at Max~$P$ ($11\%$). The qualitative ranking matches the Bayesian hierarchical classification of \citet[Figure~5]{baillon2020} (their restricted ``sharply classified'' subset of 107 of 139 subjects has Status Quo $40\%$, MaxMin $31\%$, Prospect Itself $18\%$, MinMax $9\%$, EV $2\%$, X at Max~$P$ $0\%$): SQ, MaxMin, and Prospect Itself are the three dominant rules in both classifications. The per-subject MLE here is noisier without their hierarchical shrinkage and spreads a larger tail of subjects across EV and X at Max~$P$.

\paragraph{Benchmark Monte Carlo.}
For each subject, the AIC-selected rule and its fitted parameters define menu-level left-choice probabilities $\hat q_t$, which are then plugged into the MRCI Monte Carlo benchmark of Section~\ref{subsec:model_benchmark} with $B=100$ synthetic draws. The resulting subject-level $p$-values are reported in Table~\ref{tab:baillon_rd}.


\section{The BEAST Benchmark on CPC18}\label{app:beast}

This appendix gives the model specification, calibrated parameter values, and Monte Carlo procedure for the BEAST benchmark used in Section~\ref{subsec:beast_cpc18}.

\paragraph{Model.}
BEAST (Best Estimate and Sampling Tools, \citealp{erev_ert_plonsky_2017}) is a sample-based model of risky choice. Each option is evaluated as the sum of (i) the best estimate of its expected value, (ii) the mean of $\kappa$ outcomes drawn by one of four sampling tools---an unbiased sample from the lottery, a uniform (equal-weighting) sample, a sign tool, and a contingent-pessimism tool---and (iii) Gaussian noise with scale $\sigma$. Per-agent parameters are drawn uniformly from $(0,\sigma)$, $\{1,\dots,\kappa\}$, $(0,\beta)$, $(0,\theta)$, and $(0,\gamma)$. \citet[\S2.1]{plonsky_etal2025_nhb} characterize the resulting decision process as reflecting four behavioral tendencies: minimizing immediate regret, maximizing the worst outcome, maximizing the probability of the best payoff sign, and choosing the option that performs best if all outcomes were equally likely. The full description of the model is in \citet{erev_ert_plonsky_2017}; we use the official simulator without modification.

\paragraph{Calibration.}
We use the population-level calibrated values published in \citet[p.~387]{erev_ert_plonsky_2017}: $\sigma=7$, $\kappa=3$, $\beta=2.6$, $\gamma=0.5$, $\theta=1$.\footnote{The original calibration also reports an ambiguity-aversion parameter $\varphi=0.07$ which is irrelevant here because we restrict attention to the non-ambiguous menus.} BEAST is calibrated at the population level rather than fit per subject, so the same parameter values are used throughout.

\paragraph{Implementation and Monte Carlo.}
We source the official R simulator \texttt{CPC15\_BEASTsimulation()} from the \citet{plonsky_etal2025_nhb} replication package without modification. For each unique CPC18 menu we simulate $1{,}500$ BEAST agents through the full 25-trial sequence and average across trials to obtain a menu-level left-choice probability $\hat q_t$. We then run the MRCI Monte Carlo benchmark of Section~\ref{subsec:model_benchmark} with $B=500$ synthetic datasets per subject. The resulting subject-level $p$-values are reported in Table~\ref{tab:beast_cpc18}.


\section{Proofs}\label{ec:proofs}

\subsection{Proof of Proposition~\ref{prop:properties}}\label{proof:properties}

Let $T$ be the number of menus, $s_f(\mathbf y):=\frac{1}{T}\sum_{t=1}^T y_{t,f}$ and
$\mathcal F^\star(\mathbf y):=\{f\in\mathcal F:\ s_f(\mathbf y)>0\}$.

\smallskip
\noindent\emph{(i) Bounds.} Fix any feasible profile $\mathbf y\in\mathcal C(D)$ and let $s_f(\mathbf y)\ge 0$ denote the induced rule shares with $\sum_{f\in\mathcal F}s_f(\mathbf y)=1$. By definition,
\[
\mathrm{HHI}(\mathbf y)=\sum_{f\in\mathcal F} s_f(\mathbf y)^2.
\]
Upper bound: since $s_f(\mathbf y)\in[0,1]$ and the shares sum to one, we have $\sum_f s_f^2\le 1$, with equality only when one share equals $1$ and all others are $0$. Hence $\mathrm{HHI}(\mathbf y)\le 1$ for every feasible $\mathbf y$, and therefore $\mathrm{MRCI}=\max_{\mathbf y\in\mathcal C(D)}\mathrm{HHI}(\mathbf y)\le 1$. The proof for the lower bound is provided in the text.

\smallskip
\noindent\emph{(ii) $\mathrm{MRCI}=1$ iff a one--rule rationalization exists.}
If $\mathrm{MRCI}=1$, let $\mathbf y^\star$ be an optimizer. Then $\sum_f s_f(\mathbf y^\star)^2 \le \max_f s_f(\mathbf y^\star)\cdot\sum_f s_f(\mathbf y^\star)=\max_f s_f(\mathbf y^\star)\le 1$, with equality throughout requiring $s_f(\mathbf y^\star)\in\{0,\max_f s_f(\mathbf y^\star)\}$ for every $f$ and $\max_f s_f(\mathbf y^\star)=1$; combined with $\sum_f s_f(\mathbf y^\star)=1$ this forces $s_{f^\dagger}(\mathbf y^\star)=1$ for a single $f^\dagger$ and $s_f(\mathbf y^\star)=0$ otherwise. Conversely, if some $f\in\mathcal F$ satisfies $f\in\mathcal F_t^{\mathrm{act}}$ for all $t$, then assigning $f$ to every menu is feasible and yields $\mathrm{HHI}=1$.

\smallskip
\noindent\emph{(iii) Replication invariance.}
Form $D^{(k)}$ by stacking $k$ copies of $D$, so $T^{(k)}=kT$.

\emph{($\ge$)} Replicate any feasible assignment $\mathbf y\in\mathcal C(D)$ blockwise to obtain $\mathbf y^{(k)}$.
The active sets $\mathcal F_t^{\mathrm{act}}$ repeat in each block, and admissibility is also repeated across blocks, so $\mathbf y^{(k)}\in\mathcal C(D^{(k)})$.
Let $n_f=\sum_t y_{t,f}$ and $n_f^{(k)}=\sum_{t'} y^{(k)}_{t',f}=k\,n_f$. Then
\[
\mathrm{HHI}(\mathbf y^{(k)})=\sum_f\Big(\frac{n^{(k)}_f}{T^{(k)}}\Big)^2
=\sum_f\Big(\frac{k\,n_f}{k\,T}\Big)^2
=\sum_f\Big(\frac{n_f}{T}\Big)^2
=\mathrm{HHI}(\mathbf y).
\]
Maximizing over $\mathbf y$ gives $\mathrm{MRCI}(D^{(k)})\ge\mathrm{MRCI}(D)$.

\emph{($\le$)} Conversely, take any $\mathbf y'\in\mathcal C(D^{(k)})$. For each block $b\in\{1,\dots,k\}$, its restriction $\mathbf y'^{(b)}$ to the $b$-th block lies in $\mathcal C(D)$ (active sets are identical across blocks), with within-block shares $s_f^{(b)}(\mathbf y')=\frac{1}{T}\sum_{t}y'^{(b)}_{t,f}$. The overall shares satisfy $s_f(\mathbf y')=\frac{1}{k}\sum_{b=1}^k s_f^{(b)}(\mathbf y')$. Convexity of $x\mapsto x^2$ (Jensen's inequality) gives
\[
\mathrm{HHI}(\mathbf y')=\sum_f s_f(\mathbf y')^2 \le \frac{1}{k}\sum_{b=1}^k \sum_f s_f^{(b)}(\mathbf y')^2 = \frac{1}{k}\sum_{b=1}^k \mathrm{HHI}(\mathbf y'^{(b)}) \le \mathrm{MRCI}(D),
\]
since each $\mathbf y'^{(b)}\in\mathcal C(D)$. Maximizing over $\mathbf y'$ gives $\mathrm{MRCI}(D^{(k)})\le\mathrm{MRCI}(D)$. Combining the two inequalities, $\mathrm{MRCI}(D^{(k)})=\mathrm{MRCI}(D)$.

\smallskip
\noindent\emph{(iv) Forced choice under unique active rule.}
Fix $t$ with $\mathcal F_t^{\mathrm{act}}=\{f^\ast\}$. By the activation discipline, feasibility at menu $t$ requires selecting an active rule; thus any $\mathbf y\in\mathcal C(D)$ must satisfy $y_{t,f^\ast}=1$ and $y_{t,f}=0$ for all $f\neq f^\ast$. Hence every optimizer necessarily uses $f^\ast$ at $t$.

\subsection{Proof of Proposition~\ref{prop:monotone-necessary}}\label{proof:monotone}

\emph{Monotonicity.} Let $y^\star\in\mathcal C(D;\mathcal F)$ attain $\mathrm{MRCI}(D;\mathcal F)$. Define $y'\in\mathcal C(D;\mathcal F\cup\{f\})$ by $y'_{t,r}=y^\star_{t,r}$ for $r\in\mathcal F$ and $y'_{t,f}=0$ for all $t$. Then $\mathrm{HHI}(y')=\mathrm{HHI}(y^\star)$, hence the maximum under the larger feasible set is weakly higher.

\smallskip
\emph{Necessary condition for strictness.} Fix any $\mathcal F$--optimizer $y^\star$ with shares $s^\star=(s_r^\star)_{r\in\mathcal F}$ and set $s_f^\star=0$. Consider any $\mathcal F\cup\{f\}$--feasible reassignment that moves fractions $\delta_r\in[0,s_r^\star]$ from rules $r\in\mathcal F$ to $f$, producing new shares
\[
s_f'=\sum_{r}\delta_r,\qquad s_r'=s_r^\star-\delta_r\quad(r\in\mathcal F).
\]
Feasibility requires menu--level admissibility for each reassigned observation (the activation discipline). The Herfindahl change relative to $s_f^\star=0$ is
\begin{equation}\label{eq:DeltaHHI}
\Delta \mathrm{HHI}
=\Big(\sum_r\delta_r\Big)^2+\sum_r(s_r^\star-\delta_r)^2-\sum_r(s_r^\star)^2
= (\sum_r\delta_r)^2\ -\ 2\sum_r s_r^\star \delta_r\ +\ \sum_r \delta_r^2.
\end{equation}
If mass is taken from a \emph{single} donor $r$ ($\delta_r>0$, $\delta_{q\neq r}=0$), then
\[
\Delta \mathrm{HHI}
= 2\,\delta_r(\delta_r - s_r^\star)\ \le\ 0,
\]
with equality only for $\delta_r\in\{0,s_r^\star\}$ (no move or full relabeling). Hence any one--donor reassignment cannot strictly raise concentration at $s_f^\star=0$.

Therefore, a \emph{necessary} condition for strict improvement over $\mathrm{HHI}(y^\star)$ is: there exist \emph{two} rules $r\neq s$ with $s_r^\star,s_s^\star>0$ and a feasible reassignment with $\delta_r>0$ and $\delta_s>0$. Feasibility, in turn, requires that $f$ be locally admissible on at least one menu currently assigned to $r$ and on at least one menu currently assigned to $s$. If $f$ is never locally admissible then all feasible moves satisfy $\Delta \mathrm{HHI}\le 0$, and the maximum under $\mathcal F\cup\{f\}$ cannot strictly exceed that under $\mathcal F$. Since the argument applies to every $\mathcal F$--optimizer $y^\star$, the stated condition is necessary for strict inequality of the maxima. The extension from a single new rule to general $\mathcal F'\supseteq\mathcal F$ is immediate: if no single-rule extension $\mathcal F\cup\{f\}$ with $f\in\mathcal F'\setminus\mathcal F$ strictly improves the maximum, then for any $\mathcal F'$-feasible $\mathbf y'$ one can iteratively reassign observations using a new rule $f_j\in\mathcal F'\setminus\mathcal F$ to a positive-share $\mathcal F$-rule, weakly raising $\mathrm{HHI}$ by the cross-term $2\,s_{f_j}s_g\ge 0$ at each step, until the resulting assignment is $\mathcal F$-feasible with $\mathrm{HHI}\le\mathrm{MRCI}(D;\mathcal F)$; hence $\mathrm{MRCI}(D;\mathcal F')=\mathrm{MRCI}(D;\mathcal F)$. Contrapositively, $\mathrm{MRCI}(D;\mathcal F')>\mathrm{MRCI}(D;\mathcal F)$ implies the unifier property holds for some $f\in\mathcal F'\setminus\mathcal F$.

\subsection{Proof of Proposition~\ref{prop:utility_exact}}\label{proof:utility_exact}

Fix $T \geq 1$ and condition throughout on the realized deterministic menu sequence $A_T$.
Under Assumption~\ref{ass:u0}, the observed dataset $D_T$ is generated by independent Bernoulli indicators
$d_{t,T} \sim \mathrm{Bernoulli}(q_{t,T}(\theta_0))$, $t=1,\dots,T$,
and, for each $b=1,\dots,B$, the synthetic dataset $D_T^{*(b)}$ is generated independently from the
same law on the same menu sequence $A_T$. Therefore, conditional on $A_T$, the random datasets
$D_T, D_T^{*(1)},\dots,D_T^{*(B)}$
are i.i.d. Since $M(\cdot)$ is a measurable function of the dataset, the statistics
$V_{0,T}:=M(D_T)$, $V_{b,T}:=M(D_T^{*(b)})$, $b=1,\dots,B$,
are i.i.d.\ conditional on $A_T$, hence exchangeable.

Define the rank statistic
$R_T := 1+\sum_{b=1}^B \mathbf{1}\{V_{b,T}\geq V_{0,T}\}$,
so that
$\hat p^U_{T,B}(D_T;\theta_0)=R_T/(B+1)$.
Because $(V_{0,T},V_{1,T},\dots,V_{B,T})$ is exchangeable conditional on $A_T$, the observed
statistic $V_{0,T}$ is no more likely than any of the $B+1$ coordinates to occupy one of the $m$
largest positions. With the weak inequality $\mathbf{1}\{V_{b,T}\geq V_{0,T}\}$, ties are counted
against rejection, so the resulting test is conservative. Hence, for every integer
$m \in \{0,1,\dots,B+1\}$,
\[
\Pr_{\theta_0}(R_T \leq m \mid A_T) \leq \frac{m}{B+1}.
\]
Now fix $\eta \in [0,1]$ and let $m_\eta := \lfloor \eta(B+1)\rfloor$. Then
$\{\hat p^U_{T,B}(D_T;\theta_0)\leq \eta\} = \{R_T \leq m_\eta\}$,
and therefore
\[
\Pr_{\theta_0}\!\left(\hat p^U_{T,B}(D_T;\theta_0)\leq \eta \,\middle|\, A_T\right)
=
\Pr_{\theta_0}(R_T\leq m_\eta \mid A_T)
\leq
\frac{m_\eta}{B+1}
\leq
\eta.
\]

\subsection{Proof of Proposition~\ref{prop:utility_power}}\label{proof:utility_power}

Throughout, $P_T$ denotes the true law on $D_T$ conditional on $A_T$ and $Q_{T,\theta_0}$ denotes the benchmark law of Assumption~\ref{ass:u0} used to generate synthetic draws $D_T^*$. Let $p_T^U(D_T;\theta_0) := Q_{T,\theta_0}(M(D_T^*)\geq M(D_T)\mid A_T)$.
We show that $p_T^U(D_T;\theta_0)\xrightarrow{p}0$ under $P_T$.

\medskip
\noindent
\textbf{Step 1: Lower bound for the observed MRCI under the alternative.}

By hypothesis, there exists $\bar\tau>\kappa_U(\theta_0)$ such that
\begin{equation}\label{eq:prop5-lower}
P_T\!\left(M(D_T)\geq \bar\tau\right)\to 1.
\end{equation}

\medskip
\noindent
\textbf{Step 2: Upper bound for benchmark draws.}

For each synthetic dataset $D_T^*$ drawn from $Q_{T,\theta_0}$, let
$I_{t,f,T}^* := \mathbf{1}\{f \in \mathcal F_t^{\mathrm{act}}(D_T^*)\}$
and define the realized admissibility share
$C_{f,T}^* := T^{-1}\sum_{t=1}^T I_{t,f,T}^*$.
Let $c_{f,T}(\theta_0):=\E_{Q_{T,\theta_0}}[C_{f,T}^* \mid A_T]$.
For any admissible assignment $\mathbf y\in\mathcal C(D_T^*;\mathcal F)$, the share $s_f(\mathbf y)\leq C_{f,T}^*$ for every $f$, with $s_f(\mathbf y)\geq 0$ and $\sum_f s_f(\mathbf y)=1$. Therefore
\[
M(D_T^*)
=
\max_{\mathbf y\in\mathcal C(D_T^*;\mathcal F)} \sum_{f\in\mathcal F}s_f(\mathbf y)^2
\leq
K(C_T^*),
\]
where $K(c):=\max\{\sum_f s_f^2 : s_f\geq 0,\, \sum_f s_f=1,\, s_f\leq c_f\}$.

Conditional on $A_T$, under $Q_{T,\theta_0}$, for each fixed $f$ the indicators $(I_{t,f,T}^*)_{t=1}^T$ are independent (since the only randomness in $I_{t,f,T}^*$ is $d_{t,T}^*$, and these are independent across $t$ under Assumption~\ref{ass:u0}) and bounded by $1$. By a triangular-array weak law,
$C_{f,T}^* - c_{f,T}(\theta_0) \xrightarrow{p} 0$ for each $f$.
With finitely many rules and constant limits, marginal convergence in probability implies joint convergence in probability of the vector $C_T^*$, and combined with the convergence assumption $c_{f,T}(\theta_0)\to c_f(\theta_0)$ we obtain
$C_T^* \xrightarrow{p} c(\theta_0)$.
The function $K$ is the value of a parametric linear program in $s$ with right-hand sides $c$, hence continuous on the feasibility region $\{c\ge 0:\sum_f c_f\ge 1\}$; feasibility is automatic here because $C_{\textsc{A1},T}^*+C_{\textsc{A2},T}^*=1$ identically (exactly one of $\textsc{A1},\textsc{A2}$ is active at each menu). By the continuous mapping theorem,
$K(C_T^*) \xrightarrow{p} K(c(\theta_0)) =: \kappa_U(\theta_0)$.
Therefore, for every $\varepsilon>0$,
\begin{equation}\label{eq:prop5-upper}
Q_{T,\theta_0}\!\left(M(D_T^*)\geq \kappa_U(\theta_0)+\varepsilon \,\middle|\, A_T\right)\to 0.
\end{equation}

\medskip
\noindent
\textbf{Step 3: Separation.}

Choose $\varepsilon>0$ such that $\bar\tau-\varepsilon > \kappa_U(\theta_0)+\varepsilon$.
By \eqref{eq:prop5-lower}, $P_T(M(D_T)\geq \bar\tau)\to 1$.
On that event, $M(D_T) > \kappa_U(\theta_0)+\varepsilon$, so
\[
p_T^U(D_T;\theta_0)
\leq
Q_{T,\theta_0}\!\left(M(D_T^*)\geq \kappa_U(\theta_0)+\varepsilon \mid A_T\right)
\to 0
\]
by \eqref{eq:prop5-upper}. Therefore $p_T^U(D_T;\theta_0)\xrightarrow{p} 0$ under $P_T$ and $P_T(p_T^U(D_T;\theta_0)\leq \eta)\to 1$ for every $\eta\in(0,1)$.

\medskip
\noindent
\textbf{Step 4: Monte Carlo approximation.}

Conditional on $(D_T,A_T)$, the indicators $\mathbf{1}\{M(D_T^{*(b)})\geq M(D_T)\}$, $b=1,\dots,B_T$, are i.i.d.\ Bernoulli with mean $p_T^U(D_T;\theta_0)$ under $Q_{T,\theta_0}$. By the conditional law of large numbers,
$\hat p^U_{T,B_T}(D_T;\theta_0)-p_T^U(D_T;\theta_0)\xrightarrow{p}0$.
Since $p_T^U(D_T;\theta_0)\xrightarrow{p}0$ under $P_T$, Slutsky's theorem gives
$\hat p^U_{T,B_T}(D_T;\theta_0)\xrightarrow{p}0$ under $P_T$,
hence $P_T(\hat p^U_{T,B_T}(D_T;\theta_0)\leq \eta)\to 1$ for every $\eta\in(0,1)$.

\subsection{Proof of Proposition~\ref{prop:miqp}}\label{proof:miqp}

The feasible set of $(\mathbf y,\mathbf s)$ in the MIQP of Section~\ref{appendix:miqp} is in one-to-one correspondence with the admissible assignment set $\mathcal C(D;\mathcal F)$: constraint \eqref{R1_simple} ensures exactly one rule per menu, $y_{t,f}\le a_{t,f}$ enforces the activation discipline, and \eqref{LD} defines the shares. The objective $\sum_f s_f^2$ equals the HHI. Therefore the optimal value of the MIQP equals $\MRCI(D;\mathcal F)$.


\section{Computational Strategy}\label{sec:computational}

Computing the MRCI amounts to solving a combinatorial problem: we must select, for each menu $t$, a rule $f_t\in\mathcal F_t^{\mathrm{act}}$ so as to maximize the Herfindahl--Hirschman concentration index
\[
\mathrm{HHI}(\mathbf y)
=
\sum_{f\in\mathcal F} s_f(\mathbf y)^2,
\qquad
s_f(\mathbf y):=\frac{1}{T}\sum_{t=1}^T y_{t,f},
\]
subject to the requirement that the assigned rule at each menu satisfies the activation discipline. Here $\mathbf y=(y_{t,f})_{t,f}$ is the assignment matrix with $y_{t,f}\in\{0,1\}$ and $\sum_f y_{t,f}=1$ for all $t$, and the constraint requires that if $y_{t,f}=1$, then $f\in\mathcal F_t^{\mathrm{act}}$.

An exact formulation of this problem as a mixed--integer quadratic program (MIQP) is provided in Section~\ref{appendix:miqp}. However, maximizing a convex objective function (concentration) over discrete assignments is an NP-hard problem. As demonstrated in our benchmarking exercise (see Section~\ref{appendix:miqp}), the MIQP approach scales poorly with sample size, rendering it computationally intractable for standard datasets. We therefore employ a heuristic algorithm that exploits the structure of the objective. Since we are maximizing concentration, optimal solutions tend to load on rules that can rationalize many menus simultaneously. The algorithm below is designed to efficiently explore such high--coverage, high--concentration assignments.

\subsection{Overview of the heuristic}

The procedure has three components:
\begin{enumerate}\itemsep0.25em
    \item a pre--processing step that computes, for each $(t,f)$, whether rule $f$ is locally admissible at menu $t$;
    \item a greedy assignment step that, for a given ordering of rules, assigns rules to as many unassigned admissible menus as possible;
    \item a random--restart step that repeats the greedy pass under different rule orderings, and retains the assignment with the highest $\mathrm{HHI}$.
\end{enumerate}
We now describe each component in turn.

\subsubsection*{Phase 1: Pre--processing}

\paragraph{Local admissibility.}
For every menu $t$ and rule $f\in\mathcal F$, we compute whether $f$ is strictly discriminating at $t$ in the sense of the activation discipline (Section~\ref{subsec:activation}). Formally, we set
\[
a_{t,f}
:=\mathbf 1\{f\in\mathcal F_t^{\mathrm{act}}\},
\]
where $\mathcal F_t^{\mathrm{act}}$ is the set of rules such that, once applied to $A_t$, the perceived chosen lottery $\widetilde x_t(f)$ strictly dominates its perceived alternative. In all subsequent steps we only consider assignments $(t,f)$ with $a_{t,f}=1$. This step is done once per dataset.

\subsubsection*{Phase 2: Observation-order greedy with local refinement}

The core of the heuristic is a two-step procedure that, given a random ordering of observations, builds an initial admissible assignment and then refines it by local swaps.

\paragraph{Initial assignment.}
Draw a random permutation $\pi$ of $\{1,\dots,T\}$. Maintain a counter $n_f$ of the number of observations already assigned to rule $f$, initialized at $n_f=0$. Process observations in the order $\pi$: at observation $t$, among the rules currently admissible at $t$ (i.e., $f\in\mathcal F_t^{\mathrm{act}}$), assign the one with the largest current count $n_f$. Ties are broken by a small additive uniform noise term added to each candidate's count, drawn independently across observations. Update $n_{f^\star}\leftarrow n_{f^\star}+1$ for the assigned rule $f^\star$.

At the end of the pass, every observation has been assigned to an admissible rule (the active set $\mathcal F^{\mathrm{act}}_t$ is always non-empty because of \textsc{A1}, \textsc{A2}). Let $\mathbf y^{(0)}$ denote the resulting assignment matrix.

\paragraph{Local refinement.}
Starting from $\mathbf y^{(0)}$, iterate: for each observation $t=1,\dots,T$, let $r$ be its currently assigned rule, and consider reassigning $t$ to the admissible rule $r'\in\mathcal F^{\mathrm{act}}_t\setminus\{r\}$ with the largest count $n_{r'}$. If $n_{r'}>n_r-1$ (the count that $r$ would have after removing $t$), execute the swap and update both counters. Continue until a full pass completes with no swap. The final assignment, call it $\mathbf y^{(k)}$, is the output for this restart.

The refinement loop terminates because each accepted swap strictly increases $\sum_f n_f^2 = T^2\cdot\mathrm{HHI}(\mathbf y)$, which is bounded above by $T^2$.

\subsubsection*{Phase 3: Random restarts and selection}

The Phase 2 procedure is randomized through the observation order $\pi$ and the tie-breaking noise. To explore alternative local optima, we repeat Phase 2 over $K$ independent restarts (drawing a fresh $\pi$ and fresh tie-breakers each time):
We obtain $K$ candidate assignments $\mathbf y^{(1)},\dots,\mathbf y^{(K)}$ and define
\[
\widehat{\MRCI}(D;\mathcal F)
\ :=\
\max_{1\le k\le K} \ \mathrm{HHI}(\mathbf y^{(k)}).
\]

\subsection{Accuracy and Scalability}

The heuristic replaces combinatorial optimization with a sequence of fast linear passes. We benchmarked the heuristic against the MIQP solver on a set of 3,430 subsampled optimization tasks (see Section~\ref{appendix:miqp} for full details). The mean absolute difference between the exact and heuristic MRCI is less than $0.001$, and over $96\%$ of estimates fall within a $1\%$ tolerance of the true global optimum. The heuristic scales approximately linearly in $T\times \mid \mathcal F \mid$, whereas the exact solver's runtime increases rapidly with sample size, reaching an average of 200 seconds for datasets with $140$ observations. We therefore treat $\widehat{\MRCI}(D;\mathcal F)$ as an accurate and scalable proxy for the maximal concentration consistent with the activation discipline.


\section{Mixed Integer Quadratic Programming (MIQP)}\label{appendix:miqp}

This section develops a Mixed Integer Quadratic Programming (MIQP) approach to compute the exact value of the MRCI. We use this exact formulation to benchmark the accuracy and scalability of the heuristic optimization algorithm in Section~\ref{sec:computational}.

\subsection*{MIQP Formulation}

Let $\mathcal T=\{1,\dots,T\}$ be the set of menus, $\mathcal F$ the rule library, and
$a_{t,f}:=\mathbf 1\{f\in\mathcal F_t^{\mathrm{act}}\}$ the local admissibility indicator. The MRCI is the solution to:
\[
\max_{\mathbf y,\mathbf s}\ \sum_{f\in\mathcal F} s_f^2
\]
subject to
\begin{align}
\sum_{f\in\mathcal F} y_{t,f} &= 1, \qquad \forall\, t\in\mathcal T, \tag{R1}\label{R1_simple} \\[4pt]
y_{t,f} &\le a_{t,f}, \qquad \forall\, t\in\mathcal T,\ f\in\mathcal F, \\[4pt]
s_f &= \frac{1}{T}\sum_{t\in\mathcal T} y_{t,f}, \qquad \forall\, f\in\mathcal F, \tag{LD}\label{LD} \\[4pt]
y_{t,f} &\in \{0,1\}, \qquad \forall\, t,f.
\end{align}

\begin{proposition}\label{prop:miqp}
The optimal value of the MIQP above equals $\MRCI(D;\mathcal F)$.
\end{proposition}

The proof is given in Section~\ref{ec:proofs}.

\subsection*{Benchmarking}

We compare the exact MIQP solution against the heuristic used in the main text (Section~\ref{sec:computational}).

\paragraph{Experimental setup.} We fix the CPC18 subject and draw subsamples of the full dataset by retaining $n\in\{1,2,3,4,5\}$ out of the 25 repetitions for each choice problem. For each $(i,n)$ pair, we solve both the MIQP (using a standard solver) and the heuristic. This yields 3,430 optimization tasks ($686\ \text{subjects}\times 5\ \text{repetition levels}$).

\paragraph{Scalability.}
Figure~\ref{fig:miqp_scalability} illustrates the computational necessity of the heuristic approach. As the sample size increases from 1 to 5 observations per game, the average runtime of the heuristic is $\approx 0.3$ seconds. In contrast, the runtime for the exact MIQP scales rapidly, reaching an average of 200 seconds for $n=5$. Extrapolating this to the full dataset ($n=25$) confirms that the exact method is intractable for large-scale analysis.

\begin{figure}[htbp]
    \centering
    \includegraphics[width=0.7\linewidth]{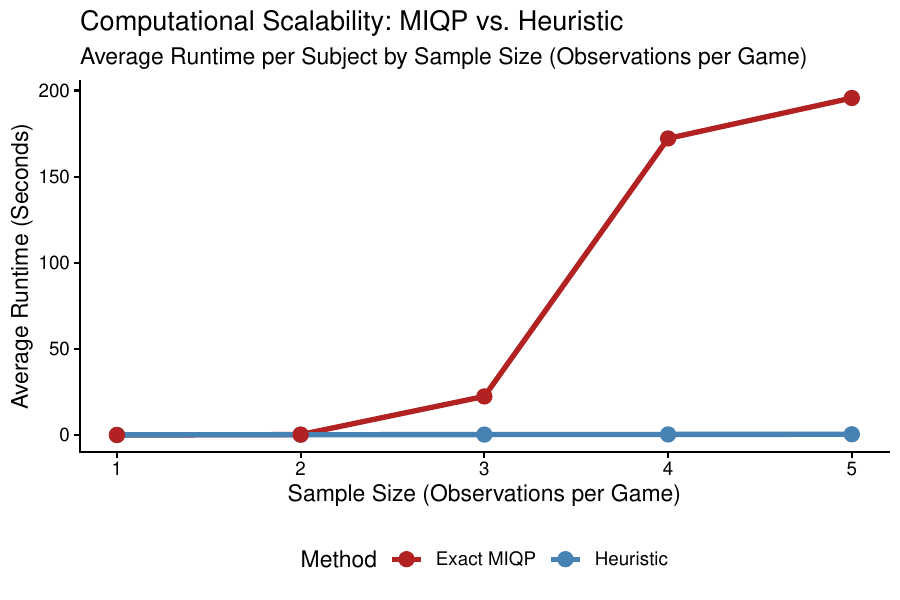}
    \caption{Scalability Benchmark. The figure shows the average computation time as a function of the number of observations per game (i.e., the number of repetitions retained from the original 25). The red line corresponds to the MIQP solver and the blue line to the heuristic.}
    \label{fig:miqp_scalability}
\end{figure}

\paragraph{Accuracy.}
We compute the absolute difference between the exact MRCI and the heuristic approximation, $\Delta = |\text{MRCI}_{\text{MIQP}} - \text{MRCI}_{\text{Heuristic}}|$. The mean absolute difference across all tasks is $9.8 \times 10^{-4}$, and $96.44\%$ of estimates fall within a strict 1\% tolerance of the true global optimum (Figure~\ref{fig:miqp_accuracy}). The maximum error observed is $0.085$. The heuristic algorithm therefore provides a high-quality approximation of the true MRCI while remaining computationally efficient for large datasets.

\begin{figure}[htbp]
    \centering
    \includegraphics[width=0.7\linewidth]{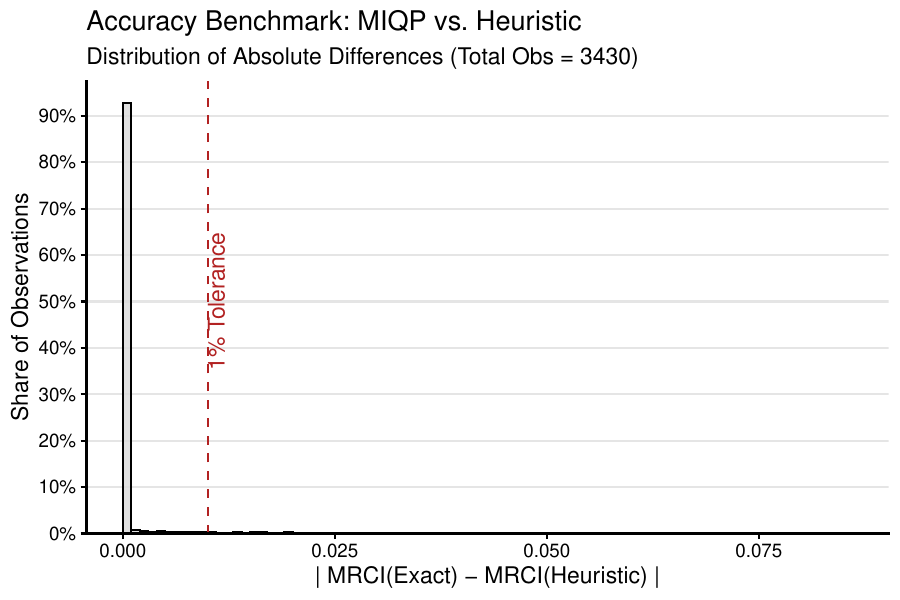}
    \caption{Accuracy Benchmark. The figure displays the distribution of the absolute difference between the exact MRCI (computed via MIQP) and the heuristic approximation. The vertical dashed line indicates a 1\% tolerance threshold.}
    \label{fig:miqp_accuracy}
\end{figure}


%% file: Table_CPT_Boundary.tex
\begin{table}[htbp]
\centering
\caption{Benchmark Rejection Rates by CPT Estimation Regime}
\label{tab:cpt_boundary}
\begin{threeparttable}
\begin{tabular}{lccccc}
\toprule
 & & \multicolumn{2}{c}{$p\le 0.05$} & \multicolumn{2}{c}{$p\le 0.01$} \\
\cmidrule(lr){3-4}\cmidrule(lr){5-6}
Subject group & $N$ & EV & CPT & EV & CPT \\
\midrule
All subjects & 686 & 0.735 & 0.615 & 0.653 & 0.529 \\
Interior $\alpha$ & 428 & 0.769 & 0.584 & 0.675 & 0.498 \\
Boundary $\alpha$ & 258 & 0.678 & 0.667 & 0.616 & 0.581 \\
Interior $\alpha$ and $\gamma$ & 378 & 0.778 & 0.587 & 0.688 & 0.505 \\
Boundary in $\alpha$ or $\gamma$ & 308 & 0.682 & 0.649 & 0.610 & 0.558 \\
\bottomrule
\end{tabular}
\begin{tablenotes}
\small
\item \textit{Notes:} $\alpha$ is the CPT value-function curvature and $\gamma$ is the Prelec probability-weighting curvature, both defined in Section~\ref{subsec:specifications} and estimated subject-by-subject. ``Interior $\alpha$'' restricts to subjects with $\hat\alpha\in(0.05,0.95)$ (off both boundaries); ``Interior $\alpha$ and $\gamma$'' additionally requires $\hat\gamma\in(0.32,1.48)$; the two ``Boundary'' rows are the complementary subsamples. Rejection rates are computed from the subject-level parametric-bootstrap $p$-values of Section~\ref{subsec:model_benchmark}.
\end{tablenotes}
\end{threeparttable}
\end{table}

%% file: Table_HeldOut_Menu.tex
\begin{table}[htbp]
\centering
\caption{Held-out-Menu Rule-Level Frontier (CPC18)}
\label{tab:heldout_menu}
\begin{threeparttable}
\begin{tabular}{cccc}
\toprule
$K$ (rules) & MRCI (held-out) & $R^{\text{test}}_K$ & Top selections (frequency on training menus) \\
\midrule
0 & 0.529 & 0.000 & --- \\
1 & 0.542 & 0.258 & SAL (60\%), MMn (20\%), MAP (18\%) \\
2 & 0.559 & 0.626 & MAP (94\%), SAL (70\%), REG (20\%) \\
3 & 0.565 & 0.753 & MAP (96\%), SAL (74\%), MMn (58\%) \\
4 & 0.572 & 0.899 & MAP (100\%), MMn (94\%), MMx (88\%) \\
5 & 0.577 & 0.998 & MAP (100\%), MMn (100\%), MMx (100\%) \\
6 & 0.577 & 1.000 & DIS (100\%), MAP (100\%), MMn (100\%) \\
\bottomrule
\end{tabular}
\begin{tablenotes}
\small
\item \textit{Notes:} Each random split divides the distinct CPC18 lottery problems into a training half and a held-out half. For each subject and each subset $S\subseteq\{\textsc{MMn},\textsc{MMx},\textsc{MAP},\textsc{SAL},\textsc{REG},\textsc{DIS}\}$, MRCI is computed on the training-problem trials and the held-out-problem trials separately. The $K$-rule subset $S^\star_K$ that maximizes cross-subject mean training-problem MRCI is selected; $R^{\text{test}}_K$ is the corresponding excess retention on held-out problems. Averaged over $50$ random splits.
\end{tablenotes}
\end{threeparttable}
\end{table}

%% file: Table_BootstrapSize.tex
\begin{table}[htbp]
\centering
\caption{Empirical Size of the Parametric-Bootstrap CPT Benchmark}
\label{tab:bootstrap_size}
\begin{threeparttable}
\begin{tabular}{ccc}
\toprule
Nominal $\eta$ & Empirical size & SE \\
\midrule
0.01 & 0.005 & 0.001 \\
0.05 & 0.029 & 0.003 \\
0.10 & 0.063 & 0.004 \\
\bottomrule
\end{tabular}
\begin{tablenotes}
\small
\item \textit{Notes:} Empirical size of the parametric-bootstrap CPT test under the null. For each of $40$ randomly sampled subjects, $M=200$ datasets were simulated from that subject's prior CPT fit on the actual menu sequence; CPT was re-estimated on each simulated dataset and a $B=100$ Monte Carlo $p$-value computed. The empirical size is the fraction of simulated datasets for which the $p$-value falls below the nominal level $\eta$.
\end{tablenotes}
\end{threeparttable}
\end{table}

%% file: Table_Rob_LeftRight.tex
\begin{table}[htbp]
\centering
\caption{Robustness: Left-Right Relabeling}
\label{tab:rob_leftright}
\begin{threeparttable}
\begin{tabular}{lc}
\toprule
Statistic & Value \\
\midrule
Mean MRCI (observed) & 0.545 \\
Mean MRCI (random relabeling) & 0.502 \\
Mean difference & 0.0430 \\
Correlation & -0.163 \\
\bottomrule
\end{tabular}
\begin{tablenotes}
\small
\item \textit{Notes:} Relabeling randomly swaps $L^1$/$L^2$ labels within each menu (prob.~0.5), destroying the systematic choice-position pairing. Mean over 20 relabelings per subject.
\end{tablenotes}
\end{threeparttable}
\end{table}

%% file: Table_Rob_Activation.tex
\begin{table}[htbp]
\centering
\caption{Robustness: Activation Discipline}
\label{tab:rob_activation}
\begin{threeparttable}
\begin{tabular}{cccc}
\toprule
$\varepsilon$ & Feasible & MRCI & Active rules/menu \\
\midrule
0.00 & 686 & 0.546 & 4.0 \\
0.10 & 686 & 0.546 & 4.0 \\
0.20 & 686 & 0.546 & 4.0 \\
0.30 & 686 & 0.546 & 4.0 \\
\bottomrule
\end{tabular}
\begin{tablenotes}
\small
\item \textit{Notes:} The activation criterion (Section~\ref{subsec:activation}) declares rule $f$ active at observation $t$ when the representative payoff $f$ assigns to the chosen lottery strictly exceeds that of the unchosen lottery. $\varepsilon=0$ is the baseline strict comparison; larger $\varepsilon$ requires the gap in representative payoffs to exceed $\varepsilon$. The attention rules \textsc{A1} and \textsc{A2} are always counted as active on the side they attend to, independent of $\varepsilon$.
\end{tablenotes}
\end{threeparttable}
\end{table}

%% file: Table_Perm_vs_CPT.tex
\begin{table}[ht!]\centering
\renewcommand{\arraystretch}{1.10}\small
\caption{Joint rejection of the permutation and parametric benchmarks at $\eta=0.01$.}
\label{tab:perm_vs_param}
\begin{threeparttable}
\begin{tabular}{lcc}
\toprule
Cell & $N$ subjects & Share \\
\midrule
\multicolumn{3}{l}{\textit{Benchmark: CPT}} \\
Both reject ($\eta=0.01$) & 281 & 41.0\% \\
Permutation only & 153 & 22.3\% \\
CPT only & 82 & 12.0\% \\
Neither & 170 & 24.8\% \\
\addlinespace[2pt]
$\Pr($CPT$\,\text{reject}\mid\text{perm reject})$ ($\eta=0.01$) & \multicolumn{2}{c}{64.7\%} \\
$\Pr($CPT$\,\text{reject}\mid\text{perm pass})$ ($\eta=0.01$) & \multicolumn{2}{c}{32.5\%} \\
\addlinespace[4pt]
\multicolumn{3}{l}{\textit{Benchmark: BEAST}} \\
Both reject ($\eta=0.01$) & 229 & 33.4\% \\
Permutation only & 205 & 29.9\% \\
BEAST only & 78 & 11.4\% \\
Neither & 174 & 25.4\% \\
\addlinespace[2pt]
$\Pr($BEAST$\,\text{reject}\mid\text{perm reject})$ ($\eta=0.01$) & \multicolumn{2}{c}{52.8\%} \\
$\Pr($BEAST$\,\text{reject}\mid\text{perm pass})$ ($\eta=0.01$) & \multicolumn{2}{c}{31.0\%} \\
\addlinespace[4pt]
\bottomrule
\end{tabular}
\begin{tablenotes}\footnotesize
\item "Permutation" denotes the i.i.d.\ benchmark of Section~\ref{subsec:permutation_special}; "CPT" the parametric bootstrap of Section~\ref{subsec:specifications}; "BEAST" (when reported) the tool-mixture benchmark of Section~\ref{subsec:beast_cpc18}. Conditional rejection rates isolate the marginal information of the parametric benchmark beyond the permutation screen.
\end{tablenotes}
\end{threeparttable}
\end{table}